\documentclass{llncs}
\usepackage[utf8]{inputenc}
\usepackage[T1]{fontenc}
\usepackage{lmodern}
\usepackage{standalone}
\usepackage{textgreek}
\usepackage{microtype}
\usepackage[bookmarks,bookmarksopen,bookmarksdepth=2,colorlinks,unicode,hypertexnames=false]{hyperref}
\usepackage{xspace}
\usepackage{float}
\usepackage{appendix}
\usepackage{pgfplots}
\usepackage{amsmath}
\usepackage{amsfonts}
\usepackage{amssymb}
\usepackage{braket}
\usepackage{mathtools}
\usepackage{siunitx}
\usepackage{enumitem}
\usepackage{algorithm}
\usepackage[noend]{algpseudocode}
\usepackage{array}
\usepackage{booktabs}
\usepackage{cleveref}
\usepackage{comment}
\usepackage{orcidlink}
\usepackage{doi}
\usepackage{color}
\usepackage{multirow}

\hypersetup{
    linkcolor=blue
    ,citecolor=blue
    ,urlcolor=blue
}

\usetikzlibrary{cd}

\setlist[enumerate,1]{label={(\arabic*)}}

\makeatletter
\newcounter{algorithmicH}
\let\oldalgorithmic\algorithmic
\renewcommand{\algorithmic}{%
    \stepcounter{algorithmicH}
    \oldalgorithmic}
\renewcommand{\theHALG@line}{ALG@line.\thealgorithmicH.\arabic{ALG@line}}
\makeatother

\allowdisplaybreaks

\spnewtheorem{thm}[theorem]{Theorem}{\bfseries}{\itshape}
\spnewtheorem{defn}[theorem]{Definition}{\bfseries}{\itshape}
\spnewtheorem{prop}[theorem]{Proposition}{\bfseries}{\itshape}
\spnewtheorem{lem}[theorem]{Lemma}{\bfseries}{\itshape}
\spnewtheorem{cor}[theorem]{Corollary}{\bfseries}{\itshape}
\spnewtheorem{hyp}[theorem]{Hypothesis}{\bfseries}{\itshape}
\spnewtheorem{conj}[theorem]{Conjecture}{\bfseries}{\itshape}
\spnewtheorem{ex}[theorem]{Example}{\bfseries}{}
\spnewtheorem{rem}[theorem]{Remark}{\bfseries}{}

\crefname{thm}{theorem}{theorems}
\Crefname{thm}{Theorem}{Theorems}
\crefname{defn}{definition}{definitions}
\Crefname{defn}{Definition}{Definitions}
\crefname{prop}{proposition}{propositions}
\Crefname{prop}{Proposition}{Propositions}
\crefname{lem}{lemma}{lemmas}
\Crefname{lem}{Lemma}{Lemmas}
\crefname{cor}{corollary}{corollaries}
\Crefname{cor}{Corollary}{Corollaries}
\crefname{ex}{example}{examples}
\Crefname{ex}{Example}{Examples}
\crefname{rem}{remark}{remarks}
\Crefname{rem}{Remark}{Remarks}
\crefname{hyp}{hypothesis}{hypotheses}
\Crefname{hyp}{Hypothesis}{Hypotheses}
\crefname{conj}{conjecture}{conjectures}
\Crefname{conj}{Conjecture}{Conjectures}

\makeatletter
\newcommand{\subalign}[1]{%
  \vcenter{%
    \Let@ \restore@math@cr \default@tag
    \baselineskip\fontdimen10 \scriptfont\tw@
    \advance\baselineskip\fontdimen12 \scriptfont\tw@
    \lineskip\thr@@\fontdimen8 \scriptfont\thr@@
    \lineskiplimit\lineskip
    \ialign{\hfil$\m@th\scriptstyle##$&$\m@th\scriptstyle{}##$\hfil\crcr
      #1\crcr
    }%
  }%
}
\makeatother

\newcolumntype{P}[1]{>{\centering\arraybackslash}p{#1}}
\newcolumntype{M}[1]{>{\centering\arraybackslash}m{#1}}

\newcommand{\F}{\mathbb{F}}                                         
\newcommand{\Fq}{\F_{q}}                                            
\newcommand{\Fqn}{\Fq^{n}}                                          

\newcommand{\abs}[1]{\left\vert #1 \right\vert}                     
\newcommand{\ceil}[1]{\left\lceil #1 \right\rceil}                  
\DeclareMathOperator{\dtv}{d_{TV}}                                  
\DeclareMathOperator{\imag}{im}                                     
\DeclareMathOperator{\lcm}{lcm}                                     
\DeclareMathOperator{\rowspace}{rowsp}                              
\ifdefined\widebar
\else
    \newcommand\widebar[1]{\mathop{\overline{#1}}}                  
\fi

\newcommand{\AffSp}[2]{\mathbb{A}^{#1}_{#2}}                        
\newcommand{\ProjSp}[2]{\mathbb{P}^{#1}_{#2}}                       

\newcommand{\homog}{\text{\normalfont hom}}                         
\DeclareMathOperator{\inid}{in}                                     
\newcommand{\satur}{\text{\normalfont sat}}                         
\DeclareMathOperator{\solvdeg}{sd}                                  
\DeclareMathOperator{\reg}{reg}                                     
\newcommand{\topcomp}{\text{\normalfont top}}                       

\DeclareMathOperator{\prob}{\mathbb{P}}                             

\newcommand*{\degree}[1]{\deg \left( #1 \right)}                    
\DeclareMathOperator{\LM}{LM}                                       
\DeclareMathOperator{\LT}{LT}                                       

\begin{document}
    \title{The Complexity of Algebraic Algorithms for {LWE}}
    \author{Matthias Johann Steiner \orcidlink{0000-0001-5206-6579}}
    \authorrunning{M.\ J.\ Steiner}
    \institute{Alpen-Adria-Universit\"at Klagenfurt, Klagenfurt am W\"orthersee, Austria \\ \email{matthias.steiner@aau.at}}

    \maketitle

    \begin{abstract}
        Arora \& Ge introduced a noise-free polynomial system to compute the secret of a Learning With Errors (LWE) instance via linearization.
        Albrecht et al.\ later utilized the Arora-Ge polynomial model to study the complexity of Gr\"obner basis computations on LWE polynomial systems under the assumption of semi-regularity.
        In this paper we revisit the Arora-Ge polynomial and prove that it satisfies a genericity condition recently introduced by Caminata \& Gorla, called being in generic coordinates.
        For polynomial systems in generic coordinates one can always estimate the complexity of DRL Gr\"obner basis computations in terms of the Castelnuovo-Mumford regularity and henceforth also via the Macaulay bound.

        Moreover, we generalize the Gr\"obner basis algorithm of Semaev \& Tenti to arbitrary polynomial systems with a finite degree of regularity.
        In particular, existence of this algorithm yields another approach to estimate the complexity of DRL Gr\"obner basis computations in terms of the degree of regularity.
        In practice, the degree of regularity of LWE polynomial systems is not known, though one can always estimate the lowest achievable degree of regularity.
        Consequently, from a designer's worst case perspective this approach yields sub-exponential complexity estimates for general, binary secret and binary error LWE.

        In recent works by Dachman-Soled et al.\ the hardness of LWE in the presence of side information was analyzed.
        Utilizing their framework we discuss how hints can be incorporated into LWE polynomial systems and how they affect the complexity of Gr\"obner basis computations.
    \end{abstract}

    \section{Introduction}
    With the emerging threat of Shor's quantum polynomial time algorithms for factoring and discrete logarithms \cite{FOCS:Shor94} on the horizon, cryptographers in the past 20 years have been in desperate search for new cryptographic problems that cannot be solved in polynomial time on classical as well as quantum computers.
    So far, lattice-based cryptography built on \emph{Learning With Errors} (LWE) and the \emph{Short Integer Solution} (SIS) \cite{STOC:Ajtai96} has emerged as most promising candidate for cryptography in the presence of quantum computers.

    In this paper we revisit polynomial models to solve the Search-LWE problem via Gr\"obner basis computations.
    Solving LWE via a polynomial system was first done by Arora \& Ge \cite{ICALP:AroGe11}, though they solved the system via linearization not via Gr\"obner bases.
    Albrecht et al.\ \cite{Albrecht-AlgebraicAlgorithms,EPRINT:ACFP14} studied the complexity of Gr\"obner basis computations for the Arora-Ge polynomial model under the assumption that the polynomial system is \emph{semi-regular} \cite{Froeberg-Conjecture,Pardue-Generic}.
    Moreover, for binary error LWE Sun et al.\ \cite{ACISP:SunTibAbe20} refined the complexity estimates for linearization under the semi-regularity assumption.
    For a general review of the computational hardness of LWE we refer to \cite{Albrecht-Hardness}.

    We stress that the complexity estimates of \cite{Albrecht-AlgebraicAlgorithms,EPRINT:ACFP14,ACISP:SunTibAbe20} are still hypothetical since both works do not provide a proof that a LWE polynomial system is semi-regular except for very special cases, see e.g.\ \cite[Theorem~11]{EPRINT:ACFP14}.
    Moreover, the complexity bounds rely on asymptotic studies of the Hilbert series of a semi-regular polynomial system.
    Needless to say that a priori is not guaranteed that these complexity estimates apply for practical LWE instantiations.

    In this paper we consider two new approaches to estimate the complexity of Gr\"obner basis computations.
    Caminata \& Gorla \cite{Caminata-SolvingPolySystems} revealed that the solving degree of polynomial system in \emph{generic coordinates} is always upper bounded by the Castelnuovo-Mumford regularity and henceforth also by the Macaulay bound, see \cite[Theorem~10]{Caminata-SolvingPolySystems}.
    For our first approach we prove that any fully determined LWE polynomial system is in generic coordinates.
    In particular this implies that for any LWE polynomial system there exists a Gr\"obner basis algorithm in exponential time as well as memory complexity.
    Semaev \& Tenti \cite{Semaev-Complexity} revealed that the complexity of Gr\"obner basis algorithms can also be estimated via the \emph{degree of regularity} of a polynomial system.
    Though, their bound is only applicable over finite fields and the polynomial system must contain the field equations, see \cite[Theorem~2.1]{Semaev-Complexity} and \cite[Theorem~3.65]{Tenti-Overdetermined}.
    We generalize their result to any polynomial system that admits a finite degree of regularity regardless of the underlying field.
    For a fixed degree of regularity we will determine the minimal number of LWE samples necessary so that the polynomial system could achieve the degree of regularity.
    Hence, for a designer this implies that there \emph{could} exist Gr\"obner basis algorithms in sub-exponential time as well as memory to solve Search-LWE.

    In two recent works Dachman-Soled et al.\ \cite{C:DDGR20,C:DGHK23} introduced a framework to study the complexity of attacks on Search-LWE in the presence of side information.
    In \Cref{Sec: hints} we shortly review their framework and describe how hints can be incorporated into LWE polynomial systems.
    Moreover, in \Cref{Ex: hints complexity} we showcase the complexity impact of hints on Gr\"obner basis computations.

    Finally, Semaev \& Tenti \cite{Semaev-Complexity} also investigated the probability that a uniformly and independently distributed polynomial system $\mathcal{F} \subset \Fq [x_1, \dots, x_n] / (x_1^q - x_1, \dots, x_n^q - x_n)$ achieves a certain degree of regularity.
    Their proof depends only on combinatorial properties, hence we expect that a similar result can be proven for uniformly and independently distributed polynomial system $\mathcal{F} \subset \Fq [x_1, \dots, x_n] / \big (f (x_1), \dots, f (x_n) \big)$, where $f$ is univariate and $\degree{f} \geq 2$ is arbitrary.
    In \Cref{Sec: Appendix} we study the related problem whether a LWE polynomial is close to the uniform distribution or not.
    We find a negative answer for this question, in particular we show that the statistical distance between the highest degree component of a LWE polynomial and the uniform distribution is always $\geq \frac{1}{2}$ and has limit $1$ if the degree of the LWE polynomial goes to infinity.
    Hence, even if Semaev \& Tenti's analysis generalizes it is not applicable to LWE polynomial systems.

    \section{Preliminaries}
    By $k$ we will always denote a field, by $\bar{k}$ we denote its algebraic closure, and by $\Fq$ we denote the finite field with $q$ elements.
    Let $I \subset k [x_1, \dots, x_n]$ be an ideal, then we denote the zero locus of $I$ over $\bar{k}$ as
    \begin{equation}
        \mathcal{Z} (I) = \left\{ \mathbf{x} \in \bar{k}^n \mid f (\mathbf{x}) = 0,\ \forall f \in I \right\} \subset \AffSp{n}{\bar{k}}.
    \end{equation}
    If in addition $I$ is homogeneous, then we denote the projective zero locus over $\bar{k}$ by $\mathcal{Z}_+ (I) \subset \ProjSp{n - 1}{\bar{k}}$.

    Let $f \in K [x_1, \dots, x_n]$ be a polynomial, and let $x_0$ be an additional variable, we call
    \begin{equation}
        f^\homog (x_0, \dots, x_n) = x_0^{\degree{f}} \cdot f \left( \frac{x_1}{x_0}, \dots, \frac{x_n}{x_0} \right) \in K [x_0, \dots, x_n]
    \end{equation}
    the homogenization of $f$ with respect to $x_0$, and analog for the homogenization of ideals $I^\homog = \left\{ f^\homog \mid f \in I \right\}$ and finite systems of polynomials $\mathcal{F}^\homog = \left\{ f_1^\homog, \dots, f_m^\homog \right\}$.
    Further, we will always assume that we can extend a term order on $k[x_1, \dots, x_n]$ to a term order on $k[x_0, \dots, x_n]$ according to \cite[Definition~8]{Caminata-SolvingPolySystems}.

    For a term order $>$ and an ideal $I \subset k [x_1, \dots, x_n]$ we denote with
    \begin{equation}
        \inid_> (I) = \{ \LT_> (f) \mid f \in I \}
    \end{equation}
    the initial ideal of $I$, i.e.\ the ideal of leading terms of $I$, with respect to $>$.

    Every polynomial $f \in [x_1, \dots, x_n]$ can be written as $f = f_d + f_{d - 1} + \ldots + f_0$, where $f_i$ is homogeneous of degree $i$.
    We denote the highest degree component $f_d$ of $f$ with $f^\topcomp$, and analog we denote $\mathcal{F}^\topcomp = \left\{ f_1^\topcomp, \dots, f_m^\topcomp \right\}$.

    For a homogeneous ideal $I \subset P$ and an integer $d \geq 0$ we denote
    \begin{equation}
        I_d = \left\{ f \in I \mid \degree{f} = d,\ f \text{ homogeneous} \right\},
    \end{equation}
    and analog for the polynomial ring $P$.

    Let $I, J \subset k [x_1, \dots, x_n]$ be ideals, then we denote with
    \begin{equation}
        I : J = \left\{ f \in k [x_1, \dots, x_n] \mid \forall g \in J \colon f \cdot g \in I \right\}
    \end{equation}
    the usual ideal quotient, and with $I : J^\infty = \bigcup_{i \geq 1} I : J^i$ the saturation of $I$ with respect to $J$.

    Let $I, \mathfrak{m} \in k [x_0, \dots, x_n]$ be homogeneous ideals where $\mathfrak{m} = (x_0, \dots, x_n)$, then we call $I^\satur = I : \mathfrak{m}^\infty$ the saturation of $I$.

    We will often encounter the lexicographic and the degree reverse lexicographic term order which we will abbreviate as LEX and DRL respectively.

    For $\mathbf{x}, \mathbf{y} \in k^n$ we denote the standard inner product as
    \begin{equation}
        \braket{\mathbf{x}, \mathbf{y}} = \mathbf{x}^\intercal \mathbf{y} = \sum_{i = 1}^{n} x_i \cdot y_i.
    \end{equation}

    By $\log$ we denote the natural logarithm and by $\log_2$ the logarithm in base $2$.

    \subsection{Learning With Errors}
    Learning With Errors (LWE) was introduced by Ajtai in his seminal work \cite{STOC:Ajtai96}.
    In its base form it can be formulated as a simple computational linear algebra problem.
    \begin{defn}[{Learning with errors, \cite{STOC:Ajtai96}}]
        Let $q$ be a prime, let $n \geq 1$ be an integer, and let $\chi$ be a probability distribution on $\mathbb{Z}$.
        For a secret vector $\mathbf{s} \in \Fqn$ the LWE distribution $A_{\mathbf{s}, \chi}$ over $\Fqn \times \Fq$ is sampled by choosing $\mathbf{a} \in \Fqn$ uniformly at random, choosing $e \leftarrow \chi$, and outputting $\left( \mathbf{a}, \braket{\mathbf{s}, \mathbf{a}} + e \in \Fq \right)$.
    \end{defn}

    In Search-LWE we are given $m$ LWE samples $(\mathbf{a}_i, b_i)$ sampled according to some probability distribution.
    Our task is then to recover the secret vector $\mathbf{s} \in \Fqn$ that has been used to generate the samples.

    As probability distribution one typically chooses a discrete Gaussian distribution with mean $0$ and standard deviation $\sigma$.
    For ease of computation in this paper, we ignore the discretization and assume $\chi = \mathcal{N} (0, \sigma)$ if not specified otherwise, hence we do not discuss discretization techniques further.
    Assume that $X \sim \mathcal{N} (0, \sigma)$, we will utilize the following well-known property of the Gaussian distribution several times in this paper
    \begin{equation}\label{Equ: Gaussian exponential decrease}
        \prob \left[ \left| X \right| > t \cdot \sigma \right] \leq \frac{2}{t \cdot \sqrt{2 \cdot \pi}} \cdot \exp \left( -\frac{t^2}{2} \right).
    \end{equation}

    It is well-known that solving Search-LWE for a discrete Gaussian error distribution and $\sigma \in \mathcal{O} \left( \sqrt{n} \right)$ is at least as hard as solving several computational lattice problems, see e.g.\ \cite{STOC:Regev05,STOC:Peikert09,STOC:BLPRS13,Micciancio-BinaryLWE}.

    Moreover, on top of LWE many cryptographic functions can be built, e.g.\ Regev's public key cryptosystem \cite{STOC:Regev05} as well as a key exchange mechanism \cite{CCS:BCDMNN16}.

    \subsection{Gr\"obner Bases}
    For an ideal $I \subset k [x_1, \dots, x_n]$ and a term order $>$ on the polynomial ring, a $>$- Gr\"obner basis $\mathcal{G} = \{ g_1, \dots, g_m \}$ is a finite set of generators such that
    \begin{equation}
        \inid_> (I) = \big( \LT_> (g_1), \dots, \LT_> (g_m) \big).
    \end{equation}
    Gr\"obner bases were introduced by Bruno Buchberger in his PhD thesis \cite{Buchberger}.
    With Gr\"obner bases one can solve many computational problems on ideals like the ideal membership problem or the computation of the zero locus \cite{Cox-Ideals}.
    For a general introduction to the theory of Gr\"obner bases we refer to \cite{Cox-Ideals}.

    Today, two classes of Gr\"obner basis algorithms are known: \emph{Buchberger's algorithm} and \emph{linear algebra-based algorithms}.
    In this paper we only study the latter family.

    Let $\mathcal{F} = \{ f_1, \dots, f_m \} \subset P = k [x_1, \dots, x_n]$ be a homogeneous polynomial system, and let $>$ be a term order on $P$.
    The \emph{homogeneous Macaulay matrix} in degree $d$, denoted as $M_d$, has columns indexed by monomials in $P_d$ sorted from left to right with respect to $>$.
    The rows of $M_d$ are indexed by polynomials $s \cdot f_i$, where $s \in P$ is a monomial such that $\degree{s \cdot f_i} = d$.
    The entry of row $s \cdot f_i$ at column $t$ is the coefficient of $s \cdot f_i$ at the monomial $t$.
    For an inhomogeneous polynomial system $M_d$ is replaced by $M_{\leq d}$ and the degree equalities by inequalities.
    By performing Gaussian elimination on $M_0, \dots, M_d$ respectively $M_{\leq d}$ for $d$ big enough one will produce a $>$-Gr\"obner basis of $\mathcal{F}$.
    This idea can be traced back to Lazard \cite{Lazard-Groebner}.
    Since $d$ determines the complexity of this algorithm in space and time, the least suitable $d$ is of special interest \cite{Ding-SolvingDegree}.
    \begin{defn}[{Solving degree, \cite[Definition~6]{Caminata-SolvingPolySystems}}]\label{Def: solving degree}
        Let $\mathcal{F} = \{ f_1, \dots, f_m \} \subset k[x_1, \dots, \allowbreak x_n]$ and let $>$ be a term order. The solving degree of $\mathcal{F}$ is the least degree $d$ such that Gaussian elimination on the Macaulay matrix $M_{\leq d}$ produces a Gr\"obner basis of $\mathcal{F}$ with respect to $>$. We denote it by $\solvdeg_> (\mathcal{F})$.

        If $\mathcal{F}$ is homogeneous, we consider the homogeneous Macaulay matrix $M_d$ and let the solving degree of $\mathcal{F}$ be the least degree $d$ such that Gaussian elimination on $M_0, \dots, M_d$ produces a Gr\"obner basis of $\mathcal{F}$ with respect to $>$.
    \end{defn}

    Today, the most efficient variants of linear algebra-based Gr\"obner basis algorithms are Faug\'ere's F4 \cite{Faugere-F4} and Matrix-F5 \cite{Faugere-F5} algorithms.
    These algorithms utilize efficient selection criteria to avoid redundant rows in the Macaulay matrices.
    Moreover, they construct the matrices for increasing values of $d$.
    Therefore, they also need stopping criteria, though one could artificially stop the computation once the solving degree is reached since then a Gr\"obner basis must already be contained in the system produced by Gaussian elimination.
    Hence, we do not discuss termination criteria further.

    Let $\mathcal{F} \subset k [x_1, \dots, x_n]$ be a polynomial system, and let $\mathcal{F}^\homog$ be its homogenization.
    We always have that, see \cite[Theorem~7]{Caminata-SolvingPolySystems},
    \begin{equation}
        \solvdeg_{DRL} \left( \mathcal{F} \right) \leq \solvdeg_{DRL} \left( \mathcal{F}^\homog \right).
    \end{equation}

    \subsubsection{Complexity Estimate via the Solving Degree.}
    For a matrix $\mathbf{A} \in k^{n \times m}$ of rank $r$ the reduced row echelon form can be computed in $\mathcal{O} \left( n \cdot m \cdot r ^{\omega - 2} \right)$ \cite[\S 2.2]{Storjohann-Matrix}, where $2 \leq \omega < 2.37286$ is a linear algebra constant \cite{SODA:AlmWil21}.

    Let $\mathcal{F} \subset P = k [x_1, \dots, x_n]$ be a system of $m$ homogeneous polynomials, it is well-known that the number of monomials in $P_d$ is given by $\binom{n + d - 1}{d}$.
    Moreover, at most $\binom{n + d - \degree{f_i} - 1}{d - \degree{f_i}} \leq \binom{n + d - 1}{d}$ many columns can stem from the polynomial $f_i$.
    Therefore, the cost of Gaussian elimination on $M_0, \dots, M_d$ is bounded by
    \begin{equation}\label{Equ: complexity estimate}
        \mathcal{O} \left( m \cdot d \cdot \binom{n + d - 1}{d}^\omega \right).
    \end{equation}
    Thus, by estimating the solving degree $\solvdeg_{DRL} (\mathcal{F})$ we yield a complexity upper bound for linear algebra-based Gr\"obner basis computations.

    \subsection{Generic Coordinates \& the Solving Degree}
    For completeness, we shortly recall the definition of the Castelnuovo-Mumford regularity \cite[Chapter~4]{Eisenbud-Syzygies}, a well-established invariant from commutative algebra and algebraic geometry.
    Let $P = k[x_0, \dots, x_n]$ be the polynomial ring and let
    \begin{equation}
        \mathbf{F}: \cdots \rightarrow F_i \rightarrow F_{i - 1} \rightarrow \cdots
    \end{equation}
    be a graded complex of free $P$-modules, where $F_i = \sum_j P(-a_{i, j})$.
    \begin{defn}
        The Castelnuovo-Mumford regularity of $\mathbf{F}$ is defined as
        \[
        \reg \left( \mathbf{F} \right) = \sup_i a_{i,j} - i.
        \]
    \end{defn}
    By Hilbert's Syzygy theorem \cite[Theorem~1.1]{Eisenbud-Syzygies} any finitely graded $P$-module has a finite free graded resolution. I.e., for every homogeneous ideal $I \subset P$ the regularity of $I$ is computable.

    Next we introduce the notion of generic coordinates which first appeared in the seminal work of Bayer \& Stillman \cite{BayerStillman}.
    Let $I \subset P$ be an ideal, and let $r \in P$.
    We use the shorthand notation ``$r \nmid 0 \mod I$'' for expressing that $r$ is not a zero-divisor on $P / I$.
    \begin{defn}[{\cite[Definition~5]{Caminata-SolvingPolySystems,Caminata-SolvingPolySystemsPreprint}}]\label{Def: generic coordinates}
        Let $k$ be an infinite field. Let $I \subset k [x_0, \dots, \allowbreak x_n]$ be a homogeneous ideal with $| \mathcal{Z}_+ (I) | < \infty$. We say that $I$ is in generic coordinates if either $| \mathcal{Z}_+ (I) | = 0$ or $x_0 \nmid 0 \mod I^{\satur}$.

        Let $k$ be any field, and let $k \subset K$ be an infinite field extension.
        $I$ is in generic coordinates over $K$ if $I \otimes_k K [x_0, \dots, x_n] \subset K [x_0, \dots, x_n]$ is in generic coordinates.
    \end{defn}

    Provided a polynomial system is in generic coordinates, then the solving degree is always upper bounded by the Castelnuovo-Mumford regularity.
    \begin{thm}[{\cite[Theorem~9, 10]{Caminata-SolvingPolySystems}}]\label{Th: solvdeg and CM-regularity}
        Let $K$ be an algebraically closed field, and let $\mathcal{F} = \{ f_1, \dots, f_m \} \subset K [x_1, \dots, x_n]$ be an inhomogeneous polynomial system such that $\left( \mathcal{F}^\homog \right)$ is in generic coordinates.
        Then
        \begin{equation*}
            \solvdeg_{DRL} \left( \mathcal{F} \right) \leq \reg \left( \mathcal{F}^\homog \right).
        \end{equation*}
    \end{thm}

    By a classical result one can always bound the regularity of an ideal with the Macaulay bound (see \cite[Theorem~1.12.4]{Chardin-Regularity}).
    \begin{cor}[{Macaulay bound, \cite[Theorem~2]{Lazard-Groebner}, \cite[Corollary~2]{Caminata-SolvingPolySystems}}]\label{Cor: Macauly bound}
        Consider a system of equations $\mathcal{F} = \{ f_1, \dots, f_m \} \subset k[x_1, \dots, x_n]$ with $d_i = \deg \left( f_i \right)$ and $d_1 \geq \ldots \geq d_m$. Set $l = \min \{ n + 1, m \}$.
        Assume that $\left| \mathcal{Z}_+ \left( \mathcal{F}^\homog \right) \right| < \infty$ and that $\left( F^\homog \right)$ is in generic coordinates over $\bar{k}$.
        Then
        \[
        \solvdeg_{DRL} \left( \mathcal{F} \right) \leq \reg \left( \mathcal{F}^\homog \right) \leq d_1 + \ldots + d_l - l + 1.
        \]
        In particular, if $m > n$ and $d = d_1$, then
        \[
        \solvdeg_{DRL} \left( \mathcal{F} \right) \leq (n + 1) \cdot (d - 1) + 1.
        \]
    \end{cor}

    In the proof of \cite[Theorem~11]{Caminata-SolvingPolySystems} Caminata \& Gorla implicitly revealed an efficient criterion to prove that a polynomial system is in generic coordinates.
    This observation was later formalized by Steiner in terms of the highest degree components of a polynomial system \cite{Steiner-Solving}.
    \begin{thm}[{\cite[Theorem~3.2]{Steiner-Solving}}]\label{Th: generic generators and highest degree components}
        Let $k$ be an algebraically closed field, and let $\mathcal{F} = \{ f_1, \dots, f_m \} \subset k [x_1, \dots, x_n]$ be an inhomogeneous polynomial system such that
        \begin{enumerate}[label=(\roman*)]
            \item $(\mathcal{F}) \neq (1)$, and

            \item $\dim \left( \mathcal{F} \right) = 0$.
        \end{enumerate}
        Then the following are equivalent.
        \begin{enumerate}
            \item $\left( \mathcal{F}^\homog \right)$ is in generic coordinates and $\left| \mathcal{Z}_+ \left( \mathcal{F}^\homog \right) \right| \neq 0$.

            \item $\sqrt{\mathcal{F}^\topcomp} = \left( x_1, \dots, x_n \right)$.

            \item $\left( \mathcal{F}^\topcomp \right)$ is zero-dimensional in $k [x_1, \dots, x_n]$.

            \item\label{Item: zero-dimensional} For every $1 \leq i \leq n$ there exists an integer $d_i \geq 1$ such that $x_i^{d_i} \in \inid_{DRL} \left( \mathcal{F}^\homog \right)$.
        \end{enumerate}
    \end{thm}

    In particular, \ref{Item: zero-dimensional} implies that every inhomogeneous polynomial system that contains a zero-dimensional DRL Gr\"obner basis is already in generic coordinates.

    \subsection{A Refined Solving Degree}
    In the Gr\"obner basis complexity literature there is another quantity that is also known as solving degree that refines \Cref{Def: solving degree}, cf.\ \cite[\S 1]{Caminata-Degrees}.
    Again let $\mathcal{F} = \{ f_1, \dots, f_m \} \subset P = k [x_1, \dots, x_n]$ be a finite set of polynomials, and let $>$ be a term order on $P$.
    We start with $M_{\leq d}$ the Macaulay matrix for $\mathcal{F}$ up to degree $d$ and compute a basis $\mathcal{B}$ of the row space of $M_{\leq d}$ via Gaussian elimination.
    Now we construct the Macaulay matrix $M_{\leq d}$ for the polynomial system $\mathcal{B}$ and again compute the basis $\mathcal{B}'$ of the row space via Gaussian elimination.
    We repeat this procedure until $\mathcal{B} = \mathcal{B}'$, at this point multiplying the polynomials in $\mathcal{B}'$ with all monomials up to degree $\leq d$ does not add any new elements to the basis after Gaussian elimination.
    We denote the final Macaulay matrix for $\mathcal{F}$ with $\hat{M}_d$, and we also denote $\hat{M}_d$'s row space via $\rowspace \left( \hat{M}_d \right)$.
    It is clear that
    \begin{equation}
        \rowspace \left( \hat{M}_d \right) \subset (\mathcal{F})_{\leq d} = \left\{ f \in (\mathcal{F}) \mid \degree{f} \leq d \right\},
    \end{equation}
    and for $d$ big enough $\rowspace \left( \hat{M}_d \right)$ will contain a $>$-Gr\"obner basis for $\mathcal{F}$.
    This motivates the following definition.
    \begin{defn}[{Refined solving degree, see \cite[Definition~1.1]{Caminata-Degrees}}]
        Let $\mathcal{F} = \{ f_1, \dots, f_m \} \subset k[x_1, \dots, \allowbreak x_n]$ and let $>$ be a term order.
        The refined solving degree of $\mathcal{F}$ is the least degree $d$ such that $\rowspace \left( \hat{M}_d \right)$ contains a Gr\"obner basis of $\mathcal{F}$ with respect to $>$.
        We denote it by $\overline{\solvdeg}_> (\mathcal{F})$.
    \end{defn}

    It is clear from the definitions that
    \begin{equation}
        \overline{\solvdeg}_> \left( \mathcal{F} \right) \leq \solvdeg_> \left( \mathcal{F} \right),
    \end{equation}
    but the inequality might be strict.

    \subsubsection{Complexity Estimate via the Refined Solving Degree.}
    Let $\mathcal{F} \subset P = k [x_1, \dots, x_n]$ be a system of $m$ homogeneous polynomials, let $\overline{\solvdeg}_> (\mathcal{F}) \leq d$ for some term order $>$ on $P$, and let $D$ denote the number of monomials in $P$ of degree $\leq d$.
    Then the dimensions of the Macaulay matrix $M_{\leq d}$ for $\mathcal{F}$ are bounded by $D \cdot m \times D$.
    Without loss of generality we can assume that $\mathcal{F}$ does not contain redundant elements, then the row space basis of $M_{\leq d}$ has either at least $m + 1$ elements or it contains a Gr\"obner basis with $\leq m$ many elements.
    In the first case, we have to build a new Macaulay matrix whose size is bounded by $D \cdot (m + 1) \times D$.
    Iterating this argument we can build at most $(D - m)$ many Macaulay matrices, and we have to perform Gaussian elimination at most $D - m$ times.
    With $D \leq d \cdot \binom{n + d - 1}{d}$ and our estimation from \Cref{Equ: complexity estimate} we obtain the following worst case complexity estimate
    \begin{align}
        &\mathcal{O} \left( \sum_{i = 0}^{D - m - 1} (m + i) \cdot d \cdot \binom{n + d - 1}{d}^\omega \right) \\
        &\in \mathcal{O} \left( \left( m \cdot D + \frac{(D - m - 1) \cdot (D - m - 2)}{2} \right) \cdot d \cdot \binom{n + d - 1}{d}^\omega \right) \\
        &\in \mathcal{O} \left( m \cdot D^2 \cdot d \cdot \binom{n + d - 1}{d}^\omega \right) \\
        &\in \mathcal{O} \left( m \cdot d^3 \cdot \binom{n + d - 1}{d}^{\omega + 2} \right). \label{Equ: refined complexity estimate}
    \end{align}

    \subsection{Approximation of Binomial Coefficients}
    We recall the following well-known approximation of binomial coefficients.
    \begin{lem}[{\cite[Lemma~17.5.1]{Cover-InformationTheory}}]
        For $0 < p < 1$, $q = 1 - p$ such that $n \cdot p$ is an integer
        \begin{equation*}
            \frac{1}{\sqrt{8 \cdot n \cdot p \cdot q}} \leq \binom{n}{n \cdot p} \cdot 2^{-n \cdot H_2 (p)} \leq \frac{1}{\sqrt{\pi \cdot n \cdot p \cdot q}}.
        \end{equation*}
    \end{lem}
    With $p = \frac{k}{n}$ the inequality then becomes
    \begin{equation}\label{Equ: binomial coefficient approximation}
        \sqrt{\frac{n}{8 \cdot k \cdot \left( n - k \right)}} \leq \binom{n}{k} \cdot 2^{-n \cdot H_2 \left( \frac{k}{n} \right)} \leq \sqrt{\frac{n}{\pi \cdot k \cdot \left( n - k \right)}}.
    \end{equation}

    In case the solving degree is an integer polynomial in the number of variables, then we have the following generic estimation for the binomial coefficient.
    \begin{prop}\label{Prop: general complexity estimation}
        Let $n \geq 2$ be an integer, let $\alpha \geq 1$, and let $p \in \mathbb{Z} [x]$.
        \begin{enumerate}
            \item If $p (n) \geq n - 1$ for all $n \geq 2$, then
            \begin{equation*}
                \left( \frac{n + p (n) - 1}{p (n) \cdot (n - 1)} \right)^\alpha \leq \frac{2^\alpha}{n - 1}.
            \end{equation*}

            \item If $p (n) \geq 0$ for all $n \geq 2$, then
            \begin{equation*}
                H_2 \left( \frac{p (n)}{n + p (n) - 1} \right) \leq \left( 4 \cdot \frac{(n - 1) \cdot p (n)}{(n + p (n) - 1)^2} \right)^\frac{1}{\log \left( 4 \right)} \leq \left( 4 \cdot \frac{p (n)}{n - 1} \right)^\frac{1}{\log \left( 4 \right)}.
            \end{equation*}
        \end{enumerate}
        In particular if $\alpha \geq 2$ and $p (n) \geq n - 1$ for all $n \geq 2$, then
        \begin{equation*}
            \binom{n + p (n) - 1}{p (n)}^\alpha \in \mathcal{O} \left( \frac{1}{n - 1} \cdot 2^{\alpha \cdot \left( \frac{4 \cdot (n - 1) \cdot p (n)}{\big( n + p (n) - 1 \big)^{2 - \log \left( 4 \right)}} \right)^\frac{1}{\log \left( 4 \right)}} \right).
        \end{equation*}
    \end{prop}
    \begin{proof}
        For (1), since $\alpha \geq 1$ and $n \geq 2$ we have that
        \begin{equation*}
            \left( \frac{n + p (n) - 1}{p (n) \cdot (n - 1)} \right)^\alpha
            = \left( \frac{1}{p (n)} + \frac{1}{n - 1} \right)^\alpha
            \leq \left( \frac{2}{n - 1} \right)^\alpha
            \leq \frac{2^\alpha}{n - 1},
        \end{equation*}
        which proves the claim.

        For (2), let $0 < p < 1$ we recall the following inequality for the binary entropy \cite[Theorem~1.2]{Topsoe-BinaryEntropy}
        \begin{equation*}
            H_2 (p) \leq \big( 4 \cdot p \cdot (1 - p) \big)^\frac{1}{\log \left( 4 \right)}.
        \end{equation*}
        Then
        \begin{align*}
            H_2 \left( \frac{p (n)}{n + p (n) - 1} \right) \leq \left( 4 \cdot \frac{(n - 1) \cdot p (n)}{(n + p (n) - 1)^2} \right)^\frac{1}{\log \left( 4 \right)}.
        \end{align*}
        Since $\log \left( 4 \right) \approx 1.3863$ we have that $n - 1 \leq n + p (n) - 1 \Rightarrow (n - 1)^\frac{1}{\log \left( 4 \right)} \leq (n + p (n) - 1)^\frac{1}{\log \left( 4 \right)}$, so the second inequality follows.

        The last claim follows from \Cref{Equ: binomial coefficient approximation} combined with the two inequalities. \qed
    \end{proof}

    \section{Refined Solving Degree \& Degree of Regularity}
    Another measure to estimate the complexity of linear algebra-based Gr\"obner basis algorithms is the so-called degree of regularity.
    \begin{defn}[{Degree of regularity, \cite[Definition~4]{Bardet-Complexity}}]
        Let $k$ be a field, and let $\mathcal{F} \subset P = k [x_1, \dots, x_n]$.
        Assume that $\left( \mathcal{F}^\topcomp \right)_d = P_d$ for some integer $d \geq 0$.
        The degree of regularity is defined as
        \begin{equation*}
            d_{\reg} \left( \mathcal{F} \right) = \min \left\{ d \geq 0 \; \middle\vert \; \left( \mathcal{F}^\topcomp \right)_d = P_d \right\}.
        \end{equation*}
    \end{defn}

    Note that by \Cref{Th: generic generators and highest degree components} and the projective weak Nullstellensatz \cite[Chapter~8~\S 3~Theorem~8]{Cox-Ideals} $\mathcal{F}$ is in generic coordinates if and only if $d_{\reg} (\mathcal{F}) < \infty$.

    Let $\mathcal{F} = \left\{ f_1, \dots, f_m, x_1^q - x_1, \dots, x_n^q - x_n \right\} \subset \Fq [x_1, \dots, x_n]$ be a polynomial system such that $d_{\reg} \left( \mathcal{F} \right) \geq \max \{ q, \degree{f_1}, \dots, \degree{f_m} \}$, Semaev \& Tenti \cite[Theorem~2.1]{Semaev-Complexity} showed that all S-polynomials appearing in Buchberger's algorithm have degree $\leq 2 \cdot d_{\reg} \left( \mathcal{F} \right) - 2$.
    Due to the requirement $d_{\reg} \left( \mathcal{F} \right) \geq q$ we do not expect that Semaev \& Tenti's bound outperforms the Macaulay bound in practice.
    On the other hand, the inclusion of the field equations was only made to restrict to the $\Fq$-valued solutions of a polynomial system, the proof of \cite[Theorem~2.1]{Semaev-Complexity} only requires that $d_{\reg} \left( \mathcal{F} \right) < \infty$.
    Moreover, we will see that LWE polynomial systems contain a univariate polynomial $f_i \mid x_i^q - x_i$ for all variables $x_i$.
    Hence, LWE polynomial systems can restrict to the $\Fq$-valued solutions with polynomials of much smaller degrees than $q$.
    Therefore, we will now generalize \cite[Theorem~2.1]{Semaev-Complexity} to the general case $d_{\reg} \left( \mathcal{F} \right) < \infty$.

    Let $\mathcal{F} = \{ f_1, \dots, f_m \} \subset k [x_1, \dots, x_n]$ be such that $d_{\reg} \left( \mathcal{F} \right) < \infty$.
    Moreover, let $>$ be a degree compatible\footnote{A term order $>$ on $P$ is called degree compatible if for $f, g \in P$ with $\degree{f} > \degree{f}$ one also has that $f > g$.} term order on $k [x_1, \dots, x_n]$.
    In principle, we simply repeat the refined analysis presented in \cite[\S 3.4]{Tenti-Overdetermined}:
    \begin{enumerate}
        \item Compute the Macaulay matrices $M_{\leq d_{\reg} (\mathcal{F})}$ of the sequence $f_1, \dots, f_m$ with respect to $>$, and put the matrix into row echelon form.

        \item Choose a finite set of generators $(\mathcal{B}) = I$ such that every element of $\mathcal{B}$ has degree $\leq d_{\reg} (\mathcal{F})$, and every monomial in $k [x_1, \dots, x_n]$ of degree $\geq d_{\reg} (\mathcal{F})$ is divisible by at least one monomial in $\big( \LM_> (\mathcal{B}) \big)$.
        \footnote{
            For ease of writing we introduce the shorthand notation: $\mathcal{B} = \{ h_1, \dots, h_r \}$, then $\big( \LM_> (\mathcal{B}) \big) = \big( \LM_> (h_1), \dots, \LM_> (h_r) \big)$.
        }
        Then we perform Buchberger's algorithm on $\mathcal{B}$ to obtain a Gr\"obner basis $\mathcal{G}$.

        \item Compute a reduced Gr\"obner basis of $(\mathcal{F})$ via $\mathcal{G}$.
    \end{enumerate}

    Let us now collect some properties of the basis $\mathcal{B}$.
    \begin{prop}\label{Prop: degree of regularity ideal basis}
        Let $k$ be a field, let $>$ be a degree compatible term order on $P = k [x_1, \dots, x_n]$, and let $\mathcal{F} = \{ f_1, \dots, f_m \} \subset P$ be such that $d_{\reg} \left( \mathcal{F} \right) < \infty$.
        There exists a finite generating set $\mathcal{B}$ for $(\mathcal{F})$ such that
        \begin{enumerate}
            \item $\max_{f \in \mathcal B} \degree{f} \leq d_{\reg} \left( \mathcal{F} \right)$.

            \item Every monomial $m \in k [x_1, \dots, x_n]$ with $\degree{m} \geq d_{\reg} \left( \mathcal{F} \right)$ is divisible by some $\LM_> (f)$, where $f \in \mathcal{B}$.

            \item For $f \in \mathcal{B}$ with $\degree{f} = d_{\reg} (\mathcal{F})$ one has $\deg \big( f - \LT_> (f) \big) < d_{\reg} (\mathcal{F})$.
        \end{enumerate}
    \end{prop}
    \begin{proof}
        We abbreviate $d_{\reg} \left( \mathcal{F} \right) = d_{\reg}$.
        First we construct the Macaulay matrix $M_{\leq d_{\reg}}$ of $\mathcal{F}$ with respect to $>$ and denote with $\mathcal{B}$ basis of the row space of $M_{\leq d_{\reg}}$.
        By assumption, we have that $d_{\reg} = d_{\reg} \left( \mathcal{B} \right)$.

        For $f \in \mathcal{F}$, if $\degree{f} \leq d_{\reg}$, then by construction $f \in (\mathcal{B})_{\leq d_{\reg}}$.
        If $\degree{f} > d_{\reg}$, then we compute the remainder $r_f$ of $f$ modulo $\mathcal{B}$ with respect to $>$ and add it to $\mathcal{B}$.
        By elementary properties of multivariate polynomial division, see \cite[Chapter~2~\S 3~Theorem~3]{Cox-Ideals}, and the degree of regularity we then have that $\degree{r_f} < d_{\reg}$.

        Obviously, we have that $(\mathcal{B}) = (\mathcal{F})$ and (1) follows by construction, (2) follows from $d_{\reg} = d_{\reg} \left( \mathcal{B} \right)$, and lastly basis elements that satisfy (3) can always be constructed with another round of Gaussian elimination on the elements of $\mathcal{B}$ of degree $d_{\reg}$. \qed
    \end{proof}

    Now we can prove the generalization of Semaev \& Tenti's bound.
    \begin{thm}\label{Th: refined complexity estimate}
        Let $k$ be a field, let $>$ be a degree compatible term order on $P = k [x_1, \dots, x_n]$, and let $\mathcal{F} = \{ f_1, \dots, f_m \} \subset P$ such that $d_{\reg} \left( \mathcal{F} \right) < \infty$.
        If $d_{\reg} \left( \mathcal{F} \right) \geq \max \big\{ \degree{f_1}, \dots, \degree{f_m} \big\}$, then
        \begin{equation*}
            \overline{\solvdeg}_> \left( \mathcal{F} \right) \leq 2 \cdot d_{\reg} \left( \mathcal{F} \right) - 1.
        \end{equation*}
    \end{thm}
    \begin{proof}
        We abbreviate $d_{\reg} (\mathcal{F}) = d_{\reg}$.
        Let $\mathcal{B} = \{ g_1, \dots, g_t \}$ be the ideal basis from \Cref{Prop: degree of regularity ideal basis} for $(\mathcal{F})$.
        By assumption, we have that $\mathcal{F} \subset (\mathcal{B})_{\leq d_{\reg}}$, and by construction $\mathcal{B} \subset \rowspace \big( M_{d_{\reg}} (\mathcal{F}) \big)$.
        Starting from $\mathcal{B}$ we compute a $>$-Gr\"obner basis via Buchberger's algorithm, see \cite[Chapter~2~\S 7]{Cox-Ideals}.
        Let $g_i, g_j \in \mathcal{B}$, we consider their $>$-S-polynomial
        \begin{equation*}
            S_> (g_i, g_j) = \frac{x^\gamma}{\LM_> (g_i)} \cdot g_i - \frac{x^\gamma}{\LM_> (g_j)} \cdot g_j,
        \end{equation*}
        where $x^\gamma = \lcm \big( \LM_> (g_i), \LM_> (g_j) \big)$.
        Note that by \cite[Chapter~2~\S 9~Proposition~4]{Cox-Ideals} we only have to consider the pairs with $\gcd \big( \LM_> (g_i), \LM_> (g_j) \big) \neq 1$.
        Since $\LM_> (g_i)$ and $\LM_> (g_j)$ must coincide in at least one variable and their degree is $\leq d_{\reg}$ we can conclude that
        \begin{equation*}
            \degree{\frac{x^\gamma}{\LM_> (g_i)} \cdot g_i}, \degree{\frac{x^\gamma}{\LM_> (g_j)} \cdot g_j} \leq 2 \cdot d_{\reg} - 1.
        \end{equation*}
        After performing division by remainder of the S-polynomial with respect to $\mathcal{B}$ we then also have that the remainder has degree $< d_{\reg}$ since $\big( \LM_> (\mathcal{B}) \big)_d = \big( k [x_1, \dots, x_n] \big)_d$ for all $d \geq d_{\reg}$.
        Therefore, we can construct all S-polynomials within Buchberger's algorithm with non-trivial remainder via polynomials whose degree is $\leq 2 \cdot d_{\reg} - 1$.
        Since Buchberger's algorithm always produces a $>$-Gr\"obner basis we can conclude that $\overline{\solvdeg}_> (\mathcal{F}) \leq 2 \cdot d_{\reg} - 1$. \qed
    \end{proof}
    \begin{cor}
        In the scenario of \Cref{Th: refined complexity estimate}, the largest degree of S-polynomials appearing in Buchberger's algorithm is less than or equal to $2 \cdot d_{\reg} \left( \mathcal{F} \right) - 2$.
    \end{cor}
    \begin{proof}
        Let us take another look at the S-polynomial
        \begin{align*}
            S_> (g_i, g_j)
            &= \frac{x^\gamma}{\LM_> (g_i)} \cdot g_i - \frac{x^\gamma}{\LM_> (g_j)} \cdot g_j \\
            &= \frac{x^\gamma}{\LM_> (g_i)} \cdot \tilde{g}_i - \frac{x^\gamma}{\LM_> (g_j)} \cdot \tilde{g}_j,
        \end{align*}
        where $x^\gamma = \lcm \big( \LM_> (g_i), \LM_> (g_j) \big)$ and $\tilde{g}_l = g_l - \LM_> (g_l)$ for $l = i, j$.
        Since the leading monomials are not coprime we have that
        \begin{equation*}
            \degree{\frac{x^\gamma}{\LM_> (g_i)}}, \degree{\frac{x^\gamma}{\LM_> (g_j)}} \leq d_{\reg} - 1.
        \end{equation*}
        Moreover, by \Cref{Prop: degree of regularity ideal basis} we have that $\degree{\tilde{g}_i}, \degree{\tilde{g}_j} < d_{\reg}$. \qed
    \end{proof}

    \section{Affine-Derived Polynomial Systems}\label{Sec: affine-derived polynomial systems}
    LWE polynomial systems follow a very special structure.
    To construct one polynomial one starts with a univariate polynomial $f$ and then substitutes a multivariate affine equation $\braket{\mathbf{a}, \mathbf{x}} + b$ into $f$.
    Many properties of LWE polynomial systems solely stem from this substitution, this motivates the following definition.
    \begin{defn}[Affine-derived polynomial systems]
        Let $k$ be a field, let $n, m \geq 1$ be integers, let $g_1, \dots, g_m \in k [x]$ be non-constant polynomials, let $\mathbf{a}_1, \dots, \mathbf{a}_m \in k^n$, and let $b_1, \dots, b_m \in k$.
        In the polynomial ring $k [x_1, \dots, x_n]$, we call
        \begin{align*}
            g_1 \left( \mathbf{a}_1^\intercal \mathbf{x} + b_1 \right) &= 0, \\
            &\dots \\
            g_m \left( \mathbf{a}_m^\intercal \mathbf{x} + b_m \right) &= 0,
        \end{align*}
        where $\mathbf{x} = \left( x_1, \dots, x_n \right)^\intercal$, the affine-derived polynomial system of $g_1, \dots, g_m$ by $(\mathbf{a}_1, b_1), \dots, (\mathbf{a}_m, b_m)$.
        We also abbreviate affine-derived polynomial systems as tuple $\Big( \big( g_1, \mathbf{a}_1, b_1 \big), \dots, \big( g_m, \mathbf{a}_m, b_m \big) \Big)$.
    \end{defn}

    Next let us collect some properties of zero-dimensional affine-derived polynomial systems.
    \begin{thm}\label{Th: affine-derived polynomial system generic generators}
        Let $k$ be a field and let $\bar{k}$ be its algebraic closure, let $n \geq 1$ be an integer, and let $\mathcal{F} = \Big( \big( g_1, \mathbf{a}_1, b_1 \big), \dots, \allowbreak \big( g_n, \mathbf{a}_n, b_n \big) \Big) \subset k [x_1, \dots, x_n]$ be an affine-derived polynomial system.
        Assume that the matrix
        \begin{equation*}
            \mathbf{A} =
            \begin{pmatrix}
                \mathbf{a}_1 & \dots & \mathbf{a}_n
            \end{pmatrix}
            ^\intercal \in k^{n \times n}
        \end{equation*}
        has rank $n$.
        Then
        \begin{enumerate}
            \item LEX and DRL Gr\"obner bases of $\mathcal{F}$ can be computed via an affine transformation.

            \item $\mathcal{F}$ is a $0$-dimensional polynomial system.

            \item $\dim_k \big( k [x_1, \dots, x_n] / (\mathcal{F}) \big) = \prod_{i = 1}^{n} \degree{g}_i$.

            \item Let $\mathcal{G} \subset \bar{k} [x_1, \dots, x_n]$ be such that $\mathcal{F} \subset \mathcal{G}$ and $(\mathcal{G}) \neq (1)$.
            Then $\left( \mathcal{G}^\homog \right)$ is in generic coordinates.
        \end{enumerate}
        If in addition $k$ is a finite field with $q$ elements, and $g_i \mid x^q - x$ for all $1 \leq i \leq n$.
        Then
        \begin{enumerate}[resume]
            \item Any ideal $I \subset k [x_1, \dots, x_n]$ such that $\mathcal{F} \subset I$ is radical.
        \end{enumerate}
    \end{thm}
    \begin{proof}
        For (1), we define new variables via
        \begin{equation*}
            \begin{pmatrix}
                y_1 \\ \vdots \\ y_n
            \end{pmatrix}
            =
            \begin{pmatrix}
                \mathbf{a}_1 & \dots & \mathbf{a}_n
            \end{pmatrix}
            ^\intercal
            \begin{pmatrix}
                x_1 \\ \vdots \\ x_n
            \end{pmatrix}
            +
            \begin{pmatrix}
                b_1 \\ \vdots \\ b_n
            \end{pmatrix}
            ,
        \end{equation*}
        and since the matrix $\mathbf{A}$ has full rank this construction is invertible.
        Then the polynomial system is of the form $g_1 (y_1) = \ldots = g_n (y_n) = 0$, so under any LEX and DRL term order the leading monomials of the polynomials are pairwise coprime, so by \cite[Chapter~2~\S 9~Theorem~3, Proposition~4]{Cox-Ideals} we have found a Gr\"obner basis.

        For (2), follows from \cite[Chapter~5~\S 3~Theorem~6]{Cox-Ideals}.

        For (3), the quotient space dimension can be computed by counting the number of monomials not contained in $\left( y_1^{\degree{g_1}}, \dots, y_n^{\degree{g_n}} \right)$.

        For (4), follows from \Cref{Th: generic generators and highest degree components}.

        For (5), let $F = \left( x_1^q - x_1, \dots, x_n^q - x_n \right) \subset k [x_1, \dots, x_n]$ be the ideal of field equations.
        It is well-known that for any ideal $I \subset k [x_1, \dots, x_n]$ the ideal $I + F$ is radical, see for example \cite[Lemma~3.1.1]{Gao-Counting}.
        Since $g_i \mid x^q - x$ we have for all $1 \leq i \leq n$ that
        \begin{equation*}
            \begin{split}
                \left( \mathbf{a}_i^\intercal \mathbf{x} + c_i \right)^q - \left( \mathbf{a}_i^\intercal \mathbf{x} + c_i \right)
                &= \left( \mathbf{a}_i^\intercal \mathbf{x} \right)^q - \left( \mathbf{a}_i^\intercal \mathbf{x} \right) \\
                &= \sum_{j = 1}^{n} a_{i, j} \cdot \left( x_j^q - x_j \right)
                = \mathbf{a}_i^\intercal
                \begin{pmatrix}
                    x_1^q - x_1 \\
                    \vdots \\
                    x_n^q - x_n
                \end{pmatrix}
                \in (\mathcal{F}).
            \end{split}
        \end{equation*}
        So by invertibility $\mathbf{A}$ we have that $x_i^q - x_i \in (\mathcal{F})$ for all $1 \leq i \leq n$ which proves the claim. \qed
    \end{proof}
    \begin{rem}
        Note that being in generic coordinates also follows from \cite[Remark~13]{Caminata-SolvingPolySystems}.
    \end{rem}

    \begin{cor}\label{Cor: solving degree}
        Let $k$ be an algebraically closed field, let $m > n \geq 1$, let $\mathcal{F} = \Big( \big( g_1, \mathbf{a}_1, b_1 \big), \dots, \allowbreak \big( g_m, \mathbf{a}_m, b_m \big) \Big) \subset k [x_1, \dots, x_n]$ be an affine-derived polynomial system such that $\degree{g_1} \geq \ldots \geq \degree{g_m}$.
        Assume that the matrix
        \begin{equation*}
            \mathbf{A} =
            \begin{pmatrix}
                \mathbf{a}_1 & \dots & \mathbf{a}_m
            \end{pmatrix}
            ^\intercal \in k^{m \times n}
        \end{equation*}
        has rank $n$.
        Then
        \begin{equation*}
            \solvdeg_{DRL} (\mathcal{F}) \leq \sum_{i = 1}^{n + 1} \left( \degree{g_i} - 1 \right) + 1.
        \end{equation*}
        In particular if $d \geq \degree{g_1}$, then
        \begin{equation*}
            \solvdeg_{DRL} (\mathcal{F}) \leq \left( n + 1 \right) \cdot \left( d - 1 \right) + 1
        \end{equation*}
    \end{cor}
    \begin{proof}
        Follows from \Cref{Th: affine-derived polynomial system generic generators} and the Macaulay bound \Cref{Cor: Macauly bound}. \qed
    \end{proof}

    \subsection{LWE Polynomial Systems}\label{Sec: general LWE}
    Arora \& Ge proposed a noise-free polynomial system to solve the Search-LWE problem \cite{ICALP:AroGe11}.
    If the error is distributed via a Gaussian distribution $\mathcal{N} (0, \sigma)$, then one assumes that the error always falls in the range $[-t \cdot \sigma, t \cdot \sigma]$ for some $t \in \mathbb{Z}$ such that $d = 2 \cdot t + 1 < q$.
    As we saw in \Cref{Equ: Gaussian exponential decrease}, the probability of falling outside this interval decreases exponentially in $t$.
    Therefore, up to some probability, in $\Fq$ the error is then always a root of the polynomial
    \begin{equation}\label{Equ: LWE polynomial}
        f (x) = x \cdot \prod_{i = 1}^{t} \left( x + i \right) \cdot \left( x - i \right) \in \Fq [x].
    \end{equation}
    Since by construction $2 \cdot t + 1 < q$ there cannot exist $1 \leq i < j \leq t$ such that $i \equiv -j \mod q$.
    So $f$ is a square-free polynomial and therefore divides the field equation $x^q - x$.
    For LWE samples $(\mathbf{a}_i, c_i) = \left( \mathbf{a}_i, \mathbf{a}_i^\intercal \mathbf{s} + e_i \right) \in \mathbb{Z}_q^n \times \mathbb{Z}_q$ one then has that in $\Fq [x_1, \dots, x_n]$
    \begin{equation}\label{Equ: LWE equation}
        f  \left( c_i - \mathbf{a}_i^\intercal \mathbf{x} \right) = 0
    \end{equation}
    with probability $\geq 1 - \frac{2}{t \cdot \sqrt{2 \cdot \pi}} \cdot \exp \left( -\frac{t^2}{2} \right)$.
    Given $m$ LWE samples one then constructs $m$ polynomials of the form of \Cref{Equ: LWE equation}, we call this polynomial system the LWE polynomial system $\mathcal{F}_\text{LWE}$.
    Obviously, the LWE polynomial system is an affine-derived polynomial system.
    The failure probability, i.e.\ the probability that at least one error term does not lie in the interval $[-t \cdot \sigma, t \cdot \sigma]$, can be estimated via the union bound
    \begin{equation}\label{Equ: failure probability}
        p_{fail} = m \cdot \prob \left[ \left| X \right| > t \cdot \sigma \right] \leq m \cdot \frac{2}{t \cdot \sqrt{2 \cdot \pi}} \cdot \exp \left( -\frac{t^2}{2} \right).
    \end{equation}
    Moreover, by \Cref{Th: affine-derived polynomial system generic generators} for the polynomial system to be fully determined we have to require that $m \geq n$ and that $n$ sample vectors are linearly independent.

    To devise the complexity of Gr\"obner basis computations we in principle follow the strategy of \cite[\S 5]{EPRINT:ACFP14}.
    We assume that $\sigma = n^\epsilon$, where $0 \leq \epsilon \leq 1$, and let $\theta$ be such that $0 \leq \theta \leq \epsilon \leq 1$.
    We consider sample numbers of the following form
    \begin{equation}
        m_\text{GB} = e^{\gamma_\theta},
    \end{equation}
    where $\gamma_\theta = 2^{2 \cdot \left( \epsilon - \theta \right)}$.
    \begin{lem}[{\cite[Lemma~5]{EPRINT:ACFP14}}]\label{Lem: success probability}
        Let $q, n, \sigma$ be parameters of an LWE instance.
        Let $\left( \mathbf{a}_1, b_1 \right), \dots, \left( \mathbf{a}_m, b_m \right)$ be elements of $\mathbb{Z}_q^n \times \mathbb{Z}$ sampled according to LWE.
        If $t = \sqrt{2 \cdot \log \left( m \right)}$, then the LWE polynomial system vanishes with probability at least
        \begin{equation*}
            p_g = 1 - \sqrt{\frac{1}{\pi \cdot \log \left( m \right)}}.
        \end{equation*}
    \end{lem}
    By \cite[Remark~1]{EPRINT:ACFP14} $m \in \mathcal{O} (n)$ implies that $p_g \in 1 - o (1)$.

    Therefore, we can deduce the degree $D_{GB}$ required for $m_\text{GB} = e^{\gamma_\theta}$ equations in the LWE polynomial system.
    By the previous lemma, we have to fix $t_\text{GB} = \sqrt{2 \cdot \log \left( m_\text{GB} \right)} = \sqrt{2 \cdot \gamma_\theta}$, so
    \begin{equation}\label{Equ: degree LWE system}
        \begin{split}
            D_\text{GB}
            &= 2 \cdot \sqrt{2 \cdot \log \left( m_\text{GB} \right)} \cdot \sigma + 1 \\
            &\in \mathcal{O} \left( \sqrt{\log \left( m_\text{GB} \right) \cdot \sigma} \right)
            = \mathcal{O} \left( \sqrt{\gamma_\theta} \cdot \sigma \right)
            = \mathcal{O} \left( n^{2 \cdot \epsilon - \theta} \right)
            = \mathcal{O} \left( \gamma_\theta \cdot n^\theta \right).
        \end{split}
    \end{equation}

    \begin{thm}\label{Th: LWE complexity bounds}
        Let $q, n \geq 2, \sigma = \sqrt{\frac{n}{2 \cdot \pi}}$ be parameters of an LWE instance.
        Let $m_\text{GB} = e^\frac{\pi \cdot n}{4}$, and let $\left( \mathbf{a}_i, b_i \right)_{1 \leq i \leq m_\text{GB}}$ be elements of $\Fqn \times \Fq$ sampled according to LWE.
        If the matrix $\mathbf{A} =
        \begin{pmatrix}
            \mathbf{a}_1 & \dots & \mathbf{a}_m
        \end{pmatrix}
        ^\intercal$ has rank $n$, then a linear algebra-based Gr\"obner basis algorithm that computes a DRL Gr\"obner basis has time complexity
        \begin{equation*}
            \mathcal{O} \left( n \cdot 2^{\omega \cdot 2^\frac{1}{\log \left( 2 \right)} \cdot n^{2 - \frac{1}{\log \left( 4 \right)}} + \frac{\pi \cdot \log_2 \left( e \right)}{4} \cdot n} \right)
        \end{equation*}
        and memory complexity
        \begin{equation*}
            \mathcal{O} \left( n \cdot 2^{2^{1 + \frac{1}{\log \left( 2 \right)}} \cdot n^{2 - \frac{1}{\log \left( 4 \right)}} + \frac{\pi \cdot \log_2 \left( e \right)}{4} \cdot n} \right).
        \end{equation*}
        The algorithm has success probability $\geq 1 - \frac{2}{\pi \cdot \sqrt{n}}$.
    \end{thm}
    \begin{proof}
        As in \Cref{Lem: success probability} let $t = \sqrt{2 \cdot \log \left( m_\text{GB} \right)}$.
        By our assumptions and \Cref{Equ: degree LWE system} we have that
        \begin{equation*}
            D_\text{GB}
            = 2 \cdot \sqrt{2 \cdot \log \left( m_\text{GB} \right)} \cdot \sigma + 1
            = 2 \cdot \sqrt{2 \cdot \frac{\pi \cdot n}{4}} \cdot \sqrt{\frac{n}{2 \cdot \pi}} + 1
            = n + 1.
        \end{equation*}
        Since the matrix $\mathbf{A}$ has full rank we can apply \Cref{Cor: solving degree} to estimate the solving degree of the LWE polynomial system
        \begin{equation*}
            \solvdeg_{DRL} \left( \mathcal{F}_\text{LWE} \right) \leq \left( n + 1 \right) \cdot \left( D_\text{GB} - 1 \right) + 1 = n^2 + n + 1.
        \end{equation*}
        Now we apply \Cref{Prop: general complexity estimation} with $p (n) = n^2 + n + 1$, then we perform the additional estimations
        \begin{align*}
            n^3 - 1 &< n^3, \\
            \left( n^2 \right)^{2 - \log \left( 4 \right)} &\leq \left( n^2 + 2 \cdot n \right)^{2 - \log \left( 4 \right)},
        \end{align*}
        for all $n \geq 1$.
        Also note that $2 - \log \left( 4 \right) \approx 0.6137$, so we can divide by the expressions in the last inequality without affecting the sign.
        Therefore,
        \begin{equation*}
            \big( n + p (n) - 1 \big) \cdot H_2 \left( \frac{p (n)}{n + p (n) - 1} \right) \leq 2^\frac{1}{\log \left( 2 \right)} \cdot n^{2 - \frac{1}{\log \left( 4 \right)}}.
        \end{equation*}
        The final claim then follows by converting $m_\text{GB}$ into base $2$. \qed
    \end{proof}

    Numerically we have that $2 - \frac{1}{\log \left( 4 \right)} \approx 1.2787$.

    \subsection{LWE With Small Errors}\label{Sec: small error LWE}
    Suppose that the LWE error distribution $\chi$ can only take values in $\mathcal{E} \subset \Fq$ with $| \mathcal{E} | = D \ll \sqrt{n}$.
    Then the error polynomial is
    \begin{equation}
        f (x) = \prod_{e \in \mathcal{E}} (x - e)
    \end{equation}
    of degree $D$.
    Moreover, for any LWE sample $(\mathbf{a}, b)$ we have $f \left( b - \mathbf{a}^\intercal \mathbf{x} \right) = 0$ with probability $1$.
    Analog to \Cref{Th: LWE complexity bounds} we can estimate the complexity of a DRL Gr\"obner basis computation.
    \begin{thm}\label{Th: small error LWE complexity bounds}
        Let $q$ be a prime, and let $m > n \geq 2$ be integers.
        Let $\left( \mathbf{a}_i, b_i \right)_{1 \leq i \leq m}$ be elements of $\Fqn \times \Fq$ sampled according to a LWE distribution $A_{\mathbf{s}, \chi}$ such that the error distribution that $\chi$ can take at most $D$ values.
        If the matrix $\mathbf{A} =
        \begin{pmatrix}
            \mathbf{a}_1 & \dots & \mathbf{a}_m
        \end{pmatrix}
        ^\intercal$ has rank $n$, then a linear algebra-based Gr\"obner basis algorithm that computes a DRL Gr\"obner basis has time complexity
        \begin{equation*}
            \mathcal{O} \left( m \cdot (D - 1) \cdot n \cdot 2^{\omega \cdot \left( 8 \cdot D^{\log \left( 4 \right) - 1} \right)^\frac{1}{\log \left( 4 \right)} \cdot n} \right)
        \end{equation*}
        and memory complexity
        \begin{equation*}
            \mathcal{O} \left( m \cdot (D - 1) \cdot n \cdot 2^{2 \cdot \left( 8 \cdot D^{\log \left( 4 \right) - 1} \right)^\frac{1}{\log \left( 4 \right)} \cdot n} \right).
        \end{equation*}
    \end{thm}
    \begin{proof}
        The LWE polynomial has degree $D$, therefore by \Cref{Cor: solving degree}
        \begin{equation*}
            \solvdeg_{DRL} (\mathcal{F}_\text{LWE}) \leq (n + 1) \cdot (D - 1) + 1.
        \end{equation*}
        We apply \Cref{Prop: general complexity estimation} with $p (n) = (n + 1) \cdot (D - 1) + 1$ and do the estimations
        \begin{equation*}
            \frac{(n + 1) \cdot (D - 1) + 1}{n - 1} = \frac{(n - 1) \cdot (D - 1) + 2 \cdot D - 1}{n - 1} \in \mathcal{O} (1),
        \end{equation*}
        for all $n \geq 2$,
        \begin{align*}
            \big( (n + 1) \cdot (D - 1) + 1 \big) \cdot (n - 1) &= \left( n^2 - 1 \right) \cdot  (D - 1) + n - 1 \leq 2 \cdot n^2 \cdot D, \\
            \left( n \cdot D \right)^{2 - \log \left( 4 \right)} &\leq \left( n \cdot D + D - 1 \right)^{2 - \log \left( 4 \right)},
        \end{align*}
        for all $n \geq 1$. \qed
    \end{proof}

    \subsection{LWE With Small Secrets}\label{Sec: small secret LWE}
    Suppose that the entries of the secret $\mathbf{s}$ of a LWE distribution $A_{\mathbf{s}, \chi}$ can only take values in $\mathcal{S} \subset \Fq$ with $| \mathcal{S} | = D$.
    Then for $1 \leq i \leq n$ we can add the equations
    \begin{equation}\label{Equ: secret equations}
        f_i (x_i) = \prod_{s \in \mathcal{S}} (x_i - s)
    \end{equation}
    to the LWE polynomial system.
    Trivially, $f_1, \dots, f_n$ is a DRL Gr\"obner basis, so the monomials $g \notin \inid_{DRL} (f_1, \dots, f_n)$ have degree $\leq n \cdot (D - 1)$.
    Moreover, any univariate polynomial is trivially affine-derived.
    \begin{thm}\label{Th: small secret LWE complexity bounds}
        Let $q$ be a prime, and let $m > n \geq 2$ be integers.
        Let $\left( \mathbf{a}_i, b_i \right)_{1 \leq i \leq m}$ be elements of $\mathbb{Z}_q^n \times \mathbb{Z}_q$ sampled according to a LWE distribution $A_{\mathbf{s}, \chi}$ such that the components of the secret can only take values in a set of size $D$.
        If the error polynomial $f$ has $\degree{f} > D$, then a linear algebra-based Gr\"obner basis algorithm that computes a DRL Gr\"obner basis has time complexity
        \begin{equation*}
            \mathcal{O} \left( m \cdot (D - 1) \cdot n^2 \cdot 2^{\omega \cdot 2^\frac{1}{\log \left( 2 \right)} \cdot (D - 1)^{1 - \frac{1}{\log \left( 4 \right)}} \cdot n^{2 - \frac{1}{\log \left( 4 \right)}}} \right)
        \end{equation*}
        and memory complexity
        \begin{equation*}
            \mathcal{O} \left( m \cdot (D - 1)^2 \cdot n^3 \cdot 2^{2^{1 + \frac{1}{\log \left( 2 \right)}} \cdot (D - 1)^{1 - \frac{1}{\log \left( 4 \right)}} \cdot n^{2 - \frac{1}{\log \left( 4 \right)}}} \right).
        \end{equation*}
    \end{thm}
    \begin{proof}
        Let $\mathcal{F}_\text{LWE}$ be the affine-derived LWE polynomial system, and let $\mathcal{F}_\mathcal{S}$ be the polynomials that have all possible values of the secret components as zeros, see \Cref{Equ: secret equations}.
        As preprocessing we compute the remainder of all polynomials in $\mathcal{F}_\text{LWE}$ with respect to $\mathcal{F}_\mathcal{S}$ and DRL, then the remainder polynomials can at most have degree $n \cdot (D - 1)$, see \cite[Chapter~2~\S 6~Proposition~1]{Cox-Ideals}.
        Now we join the remainders and $\mathcal{F}_\mathcal{S}$ in a single system $\mathcal{F}$ and start the Gr\"obner basis computation.
        By \Cref{Th: generic generators and highest degree components} this polynomial system is in generic coordinates, therefore
        \begin{equation*}
            \solvdeg_{DRL} \left( \mathcal{F} \right) \leq (n + 1) \cdot \big( n \cdot (D - 1) - 1 \big) + 1.
        \end{equation*}
        Now we apply \Cref{Prop: general complexity estimation} with $p (n) = (n + 1) \cdot n \cdot (D - 1) + 1$ and perform the additional estimations
        \begin{equation*}
            \frac{(n + 1) \cdot \big( n \cdot (D - 1) - 1 \big) + 1}{n - 1} \leq \frac{(D - 1) \cdot (n + 1)^2}{n - 1} \in \mathcal{O} \left( (D - 1) \cdot n \right)
        \end{equation*}
        for all $n \geq 2$, and
        \begin{align*}
            \big( (n + 1) \cdot n \cdot (D - 1) + 1 \big) \cdot (n - 1) &\leq n^3 \cdot (D - 1), \\
            n^2 \cdot (D - 1) &\leq n + (n + 1) \cdot n \cdot (D - 1),
        \end{align*}
        for all $n \geq 1$.
        Then
        \begin{equation*}
            \frac{n^3 \cdot (D - 1)}{\big( n^2 \cdot (D - 1) \big)^{2 - \log \left( 4 \right)}} = n^{2 \cdot \log \left( 4 \right) - 1} \cdot (D - 1)^{\log \left( 4 \right) - 1}
        \end{equation*}
        which proves the claim. \qed
    \end{proof}

    \subsubsection{LWE With Small Secrets \& Small Errors.}
    Lastly, let us shortly analyze the case of small secret small error LWE.
    Suppose that the errors are drawn from a set of size $D_\mathcal{E}$ and that the secrets are drawn from a set of size $D_\mathcal{S}$.
    As for \Cref{Th: small secret LWE complexity bounds} we can compute the DRL remainder of the LWE polynomials with respect to the $n$ univariate polynomials limiting the possible solutions for the secret.
    \begin{itemize}
        \item If $D_\mathcal{E} \gg D_\mathcal{S}$, then we can estimate the degrees of the remainders as $\leq n \cdot (D_\mathcal{S} - 1)$, then we obtain the Macaulay bound
        \begin{equation}
            \solvdeg_{DRL} \left( \mathcal{F} \right) \leq (n + 1) \cdot n \cdot (D_\mathcal{S} - 1) + 1.
        \end{equation}

        \item If $n \cdot (D_\mathcal{S} - 1) \gg D_\mathcal{E} \geq D_\mathcal{S}$, then we can always estimate the degrees of the remainders as $\leq D_\mathcal{E}$, then
        \begin{equation}
            \solvdeg_{DRL} \left( \mathcal{F} \right) \leq (n + 1) \cdot (D_\mathcal{E} - 1) + 1.
        \end{equation}
        \item If $n \cdot (D_\mathcal{S} - 1) \gg D_\mathcal{S} > D_\mathcal{E}$, then we perform a variable transformation so that the LWE polynomials $\mathcal{F}_\text{LWE}$ include $n$ univariate polynomials, i.e.\ we exchange the roles of $\mathcal{F}_\mathcal{S}$ and $\mathcal{F}_\text{LWE}$.
        The degrees of the remainders of $\mathcal{F}_\mathcal{S}$ are then bounded by $\leq D_\mathcal{S}$, and we obtain
        \begin{equation}
            \solvdeg_{DRL} \left( \mathcal{F} \right) \leq n \cdot (D_\mathcal{S} - 1) + D_\mathcal{E}.
        \end{equation}
    \end{itemize}
    So the first case reduces to \Cref{Th: small secret LWE complexity bounds} and the second and the third one to \Cref{Th: small error LWE complexity bounds}, though the third case has a different constant term in the solving degree bound than small error LWE.

    \section{Sub-Exponential Complexity Estimates via the Refined Solving Degree}\label{Sec: sub-exponential complexity estimates}
    In this section we use \Cref{Th: refined complexity estimate} to show that in an ideal scenario general LWE, binary secret LWE and binary error LWE admit sub-exponential Gr\"obner basis algorithms.

    \subsection{LWE With Exponential Many Samples}\label{Sec: LWE sub-exponential}
    For general LWE the lowest achievable degree of regularity is the degree $D$ of the error polynomial.
    In that degree there exist $\binom{n + D - 1}{D}$ many monomials, hence to achieve degree of regularity $m$ the number of samples $m$ has to be at least the aforementioned binomial coefficient.
    \begin{thm}\label{Th: LWE sub-exponential}
        Let $q, n, \sigma$ be parameters of an LWE instance, and let $D = 2 \cdot t \cdot \sigma + 1$ be the degree of the LWE polynomial.
        Let $m \in \mathcal{O} \left( \binom{n + D - 1}{D} \right)$ be such that $d_{\reg} \left( \mathcal{F}_\text{LWE}^\topcomp \right) = D$.
        Then a linear algebra-based Gr\"obner basis algorithm that computes a DRL Gr\"obner basis has time complexity
        \begin{equation*}
            \mathcal{O} \left( D^3 \cdot 2^{(\omega + 3) \cdot 2^\frac{1}{\log \left( 2 \right)} \cdot (2 \cdot D - 1)^\frac{1}{\log \left( 4 \right)} \cdot (n - 1)^{1 - \frac{1}{\log \left( 4 \right)}}} \right)
        \end{equation*}
        and memory complexity
        \begin{equation*}
            \mathcal{O} \left( D^3 \cdot 2^{5 \cdot 2^\frac{1}{\log \left( 2 \right)} \cdot (2 \cdot D - 1)^\frac{1}{\log \left( 4 \right)} \cdot (n - 1)^{1 - \frac{1}{\log \left( 4 \right)}}} \right).
        \end{equation*}
        For $t \to \infty$ the success probability of the algorithm approaches $1$.
    \end{thm}
    \begin{proof}
        We can use \Cref{Th: refined complexity estimate} and \Cref{Equ: refined complexity estimate} to estimate the complexity of a linear algebra based Gr\"obner basis algorithm.
        Then
        \begin{equation*}
            \mathcal{O} \left( m \cdot (2 \cdot D - 1)^3 \cdot \binom{n + 2 \cdot D - 2}{2 \cdot D - 1}^{\omega + 2} \right)
            \in \mathcal{O} \left( D^3 \cdot \binom{n + 2 \cdot D - 2}{2 \cdot D - 1}^{\omega + 3} \right).
        \end{equation*}
        To estimate the binomial coefficient we use \Cref{Equ: binomial coefficient approximation} and \cite[Theorem~1.2]{Topsoe-BinaryEntropy}.
        Similar to \Cref{Prop: general complexity estimation}, the term in the square root is estimated by $\mathcal{O} \left( 1 \right)$.
        For the entropy term we have that
        \begin{align*}
            (n + 2 \cdot D - 2) \cdot H_2 \left( \frac{2 \cdot D - 1}{n + 2 \cdot D - 2} \right) \leq \left( 4 \cdot \frac{(2 \cdot D - 1) \cdot (n - 1)}{(n + 2 \cdot D - 2)^{2 - \log \left( 4 \right)}} \right)^\frac{1}{\log \left( 4 \right)}.
        \end{align*}
        Without loss of generality $D \geq 1$, so $n - 1 \leq n + 2 \cdot D - 2$ which implies the complexity claim.

        For the success probability, recall that by \Cref{Equ: failure probability}
        \begin{align*}
            p_{fail}
            &\in \mathcal{O} \left( m \cdot \frac{2}{t \cdot \sqrt{2 \cdot \pi}} \cdot \exp \left( -\frac{t^2}{2} \right) \right) \\
            &\in \mathcal{O} \left( \binom{n + D - 1}{D} \cdot \frac{2}{t \cdot \sqrt{2 \cdot \pi}} \cdot \exp \left( -\frac{t^2}{2} \right) \right) \\
            &\in \mathcal{O} \left( \sqrt{\frac{n + D - 1}{D \cdot (n - 1)}} \cdot 2^{2 \cdot \sqrt{D \cdot n}} \cdot \frac{2}{t \cdot \sqrt{2 \cdot \pi}} \cdot \exp \left( -\frac{t^2}{2} \right) \right) \\
            &\in \mathcal{O} \left( \exp \left( 2 \cdot \log \left( 2 \right) \cdot \sqrt{2 \cdot t \cdot \sigma \cdot n} - \frac{t^2}{2} \right) \right),
        \end{align*}
        which proves the claim. \qed
    \end{proof}

    In particular, for $\sigma = \sqrt{n}$ and $t = \frac{k}{\sqrt{\sigma}}$, where $k \in \mathbb{Z}$ we obtain the complexity estimate
    \begin{equation}
        \mathcal{O} \left( \left( k \cdot \sqrt{n} \right)^3 \cdot 2^{(\omega + 3) \cdot 2^\frac{1}{\log \left( 2 \right)} \cdot (4 \cdot k + 1)^\frac{1}{\log \left( 4 \right)} \cdot n^{1 - \frac{1}{2 \cdot \log \left( 4 \right)}}} \right).
    \end{equation}
    Since $1 - \frac{1}{2 \cdot \log \left( 4 \right)} \approx 0.6393$ this complexity estimate is sub-exponential.

    \subsection{Sub-Exponential Complexity for Binary Secret LWE}\label{Sec: binary secret LWE sub-exponential}
    Recall that binary secret LWE is the simplest case of small secret LWE, see \Cref{Sec: small secret LWE}.
    Let $F = (x_1^2 - x_1, \dots, x_n^2 - x_n)$, and let $\mathcal{F}_\text{LWE} = \{ f_1, \dots, f_m \}$ be a binary secret LWE polynomial system where the (univariate) LWE error polynomial is of degree $D$.
    Without loss of generality we can first reduce the polynomials in $\mathcal{F}_\text{LWE}$ modulo $F$ with respect to the DRL term order.
    Let $f \in \mathcal{F}_\text{LWE}$, after the preprocessing step only monomials of the form
    \begin{equation}\label{Equ: monomials in binary secret LWE}
        m = x_1^{\alpha_1} \cdots x_n^{\alpha_n},
    \end{equation}
    where $\alpha_i \in \{ 0, 1 \}$ for all $i$, are present in $f$ and by elementary properties of multivariate polynomial division, see \cite[Chapter~2~\S 3]{Cox-Ideals}, also $\degree{f} \leq D$ after the reduction.

    Suppose that all $f \in \mathcal{F}_\text{LWE}$ are of degree $D$ after the reduction, we want to find the minimal achievable degree of regularity $d_{\reg} \big( (\mathcal{F}_\text{LWE}) + F \big)$.
    Let $g \in P = \Fq [x_1, \dots, x_n]$ be a monomial such that $x_i^2 \mid g$ for some $i$.
    Such a monomial can always be generated by some element in $F^\topcomp$, therefore we only have to consider monomials as in \Cref{Equ: monomials in binary secret LWE}.
    Necessarily, these monomials must be generated by the elements in $\mathcal{F}_\text{LWE}^\topcomp$.
    Moreover, by elementary combinatorics there exist $\binom{n}{d}$ many monomials of the form of \Cref{Equ: monomials in binary secret LWE} in degree $d$.

    To compute $d_{\reg} \big( (\mathcal{F}_\text{LWE}) + F \big)$ one iterates through:
    \begin{enumerate}
        \item Let $d = 0$, and $\mathcal{G} = \left( \mathcal{F}_\text{LWE}^\topcomp \right)$.

        \item Perform Gaussian elimination on $\mathcal{G}$ to obtain a minimal generating set.
        If $\abs{\mathcal{G}} = \binom{n}{D + d}$ return $D + d$, else set $d = d + 1$.

        \item Compute $\mathcal{G} = \sum_{i = 1}^{n} x_i \cdot (\mathcal{G}) \mod (x_1^2, \dots, x_n^2)$, and return to step (2).
    \end{enumerate}

    In order to achieve $d_{\reg} \big( (\mathcal{F}_\text{LWE}) + F \big) \leq D + d$, for some $d \geq 0$, we must require that
    \begin{align}
        m \cdot \binom{n}{d} &\stackrel{!}{\geq} \binom{n}{D + d} \\
        \Leftrightarrow m &\stackrel{!}{\geq} \frac{\binom{n}{D + d}}{\binom{n}{d}} = \prod_{i = 1}^{D} \frac{n - d - i + 1}{d + i}. \label{Equ: necessary number of samples}
    \end{align}
    I.e., $m \in \mathcal{O} \left( n^D \right)$ many samples can be sufficient to achieve $d_{\reg} \big( (\mathcal{F}_\text{LWE}) + F \big) \leq D + 1$.

    Provided that $m \in \mathcal{O} \left( n^D \right)$ and $d_{\reg} (\mathcal{F}_\text{LWE}) \leq D + 1$, then we obtain analog to \Cref{Th: LWE sub-exponential} the following complexity estimate
    \begin{equation}\label{Equ: binary secret LWE sub-exponential}
        \mathcal{O} \left( n^D \cdot D^3 \cdot 2^{(\omega + 2) \cdot 2^\frac{1}{\log \left( 2 \right)} \cdot (2 \cdot D + 1)^\frac{1}{\log \left( 4 \right)} \cdot (n - 1)^{1 - \frac{1}{\log \left( 4 \right)}}} \right).
    \end{equation}
    If $D = 2 \cdot t \cdot \sigma + 1$ and $\sigma = \sqrt{n}$, then we can further estimate $2 \cdot D + 1 \in \mathcal{O} \left( \sqrt{n} \right)$.
    In particular, the exponent of $n$ then becomes
    \begin{equation}
        \frac{1}{2 \cdot \log \left( 4 \right)} + 1 - \frac{1}{\log \left( 4 \right)} = 1 - \frac{1}{2 \cdot \log \left( 4 \right)} \approx 0.6393,
    \end{equation}
    so the complexity estimate is indeed sub-exponential.

    \subsection{Polynomial Complexity for Binary Error LWE}\label{Sec: binary error LWE sub-exponential}
    Recall that binary error LWE is the simplest case of small error LWE, see \Cref{Sec: small error LWE}.
    Every polynomial has degree $2$.
    Analog to \Cref{Th: small error LWE complexity bounds}, we first pick $n$ linearly independent samples $(\mathbf{a}_i, b_i)$ and perform a coordinate transformation.
    So without loss of generality we can assume that $\mathbf{a}_i$ is the $i$\textsuperscript{th} standard basis vector of $\Fqn$.
    After the transformation these $n$ LWE equations become $x_i^2 - x_i = 0$.
    We allocate them in the ideal $F = (x_1^2 - x_1, \dots, x_n^2 - x_n)$, the remaining $m - n$ LWE polynomials we collect in $\mathcal{F}_\text{LWE}$.
    Therefore, we can interpret binary error LWE as special case of binary secret LWE, see \Cref{Sec: binary secret LWE sub-exponential}.
    Suppose that we want to achieve $d_{\reg} \big( (\mathcal{F}_\text{LWE}) + F \big) \leq 2 + d$ for some $d \geq 0$, then by \Cref{Equ: necessary number of samples}
    \begin{equation}\label{Equ: necessary number of samples binary error LWE}
        m - n \stackrel{!}{\geq} \frac{(n - d - 1) \cdot (n - d)}{(d + 1) \cdot (d + 2)}
    \end{equation}
    many LWE samples are necessary.
    In particular, for $d = 0$ this reduces to Arora \& Ge's analysis \cite{ICALP:AroGe11}.
    Analog to \Cref{Th: LWE sub-exponential} and \Cref{Equ: binary secret LWE sub-exponential}, for $m \in \mathcal{O} \left( n^2 \right)$ we then obtain the complexity estimate
    \begin{equation}\label{Equ: binary error LWE polynomial}
        \mathcal{O} \left( n^2 \cdot d^3 \cdot \binom{n + 2 \cdot d + 2}{2 \cdot d + 3}^{\omega + 2} \right)
        \in \mathcal{O} \left( d^3 \cdot n^{(\omega + 2) \cdot (2 \cdot d + 3) + 2} \right).
    \end{equation}

    It is easy to see from \Cref{Equ: necessary number of samples binary error LWE} that the higher the value of $d$ the fewer samples are necessary to achieve a certain degree of regularity.
    Let us see an example.
    \begin{ex}\label{Ex: binary secret LWE optimal estimate}
        Let $q$ be a prime, and let $n = 256$, and
        \begin{enumerate}
            \item Let $m = 2 \cdot n$.
            The minimum $d \in \mathbb{Z}_{\geq 0}$ such that \Cref{Equ: necessary number of samples binary error LWE} is satisfied is $d = 14$.
            Analog to \Cref{Equ: binary error LWE polynomial} with $m = 2 \cdot n$ we yield the complexity of a DRL Gr\"obner basis computation
            \begin{equation*}
                \mathcal{O} \left( 2 \cdot n \cdot d^3 \cdot \binom{n + 30}{31}^{\omega + 2} \right) \in \mathcal{O} \left( n^{31 \cdot \omega + 64} \right).
            \end{equation*}
            If we use $\omega \leq 3$, then direct evaluation of the left complexity yields $434$ bits.

            \item Let $m = n^\frac{3}{2}$.
            The minimum $d \in \mathbb{Z}_{\geq 0}$ such that \Cref{Equ: necessary number of samples binary error LWE} is satisfied is $d = 3$.
            Then we yield the complexity of a DRL Gr\"obner basis computation
            \begin{equation*}
                \mathcal{O} \left( n^\frac{3}{2} \cdot d^3 \cdot \binom{n + 8}{9}^{\omega + 2} \right) \in \mathcal{O} \left( n^{9 \cdot \omega + 19.5} \right).
            \end{equation*}
            If we use $\omega \leq 3$, then direct evaluation of the left complexity yields $178$ bits.
        \end{enumerate}
    \end{ex}

    \subsection{A Conjecture on the Castelnuovo-Mumford Regularity}
    Experimentally we observed the following property for all LWE polynomial systems studied in this paper.
    \begin{conj}\label{Conj: regularity bound}
        Let $\Fq$ be a finite field, and let $\mathcal{F}_\text{LWE} \subset \Fq [x_1, \dots, x_n]$ be a LWE polynomial system.
        \begin{enumerate}
            \item For small secret LWE where the error is drawn from the interval $[-N, N]$
            \begin{equation*}
                \reg \left( \mathcal{F}_\text{LWE}^\homog \right) \leq d_{\reg} \left( \mathcal{F}_\text{LWE} \right) + N - 1.
            \end{equation*}

            \item For binary secret or binary error LWE
            \begin{equation*}
                \reg \left( \mathcal{F}_\text{LWE}^\homog \right) \leq d_{\reg} \left( \mathcal{F}_\text{LWE} \right) + 1.
            \end{equation*}
        \end{enumerate}
    \end{conj}

    In case the conjecture holds, then the complexity estimates discussed in this section improve significantly since we can utilize the complexity estimate for Gaussian elimination on a \emph{single} Macaulay matrix (\Cref{Equ: complexity estimate}).
    \begin{itemize}
        \item The binary error LWE estimate from \Cref{Equ: binary secret LWE sub-exponential} improves to
        \begin{equation}
            \mathcal{O} \left( n^D \cdot D \cdot 2^{\omega \cdot 2^\frac{1}{\log \left( 2 \right)} \cdot (D + 2)^\frac{1}{\log \left( 4 \right)} \cdot (n - 1)^{1 - \frac{1}{\log \left( 4 \right)}}} \right).
        \end{equation}

        \item The binary secret LWE estimate from \Cref{Equ: binary error LWE polynomial} improves to
        \begin{equation}
            \mathcal{O} \left( d \cdot n^{\omega \cdot (d + 3) + 2} \right).
        \end{equation}
    \end{itemize}

    E.g., under the conjecture the numeric complexities of \Cref{Ex: binary secret LWE optimal estimate} improve to $279$ bits and $96$ bits respectively.

    We also note that for the conservative cryptanalyst there is a non-hypothetical alternative to \Cref{Conj: regularity bound}.
    By \cite[Theorem~5.3]{Caminata-Degrees} for a polynomial system $\mathcal{F}^\homog \subset P [x_0]$ in generic coordinates one always has that
    \begin{equation}\label{Equ: regularity lower bound}
        d_{\reg} \left( \mathcal{F} \right) \leq \reg \left( \mathcal{F}^\homog \right).
    \end{equation}

    Thus, one can estimate the lowest achievable complexity estimate for Gaussian elimination on the Macaulay matrix to produce a Gr\"obner basis of $\mathcal{F}_\text{LWE}$ as follows:
    \begin{enumerate}
        \item Compute/Estimate the lowest achievable degree of regularity $\hat{d}$ for $\mathcal{F}_\text{LWE}$.

        \item Use \Cref{Equ: complexity estimate} with $d = \hat{d}$ and $\omega = 2$ to estimate the lowest achievable complexity upper bound of a Gr\"obner basis computation for $\mathcal{F}_\text{LWE}$.
    \end{enumerate}

    We also recommend utilizing \Cref{Equ: complexity estimate} itself for numerical computations rather than our complexity estimations.
    Our estimations are not tight but merely showcase the complexity class, i.e.\ exponential, sub-exponential \& polynomial, for various LWE Gr\"obner basis computations.

    \subsection{Complexity Estimation of Kyber768}
    Finally, let us showcase our complexity estimation methods for a concrete cryptographic example: Kyber768 \cite{NISTPQC:CRYSTALS-KYBER22}, a selected algorithm in the NIST post-quantum competition.
    Kyber768 is based on the Module-LWE problem, it has parameters $q = 3329$, $n = 3 \cdot 256$, $m = n$, $D = 2$ and errors as well as secrets are drawn from the interval $[-D, D]$.
    I.e., it is an instance of small error and small secret LWE.
    Thus, it induces a polynomial system of $1536$ equations in $768$ variables, where $768$ polynomials stem from LWE samples.
    The lowest achievable degree of regularity for Kyber768 is estimated via
    \begin{equation}
        m \cdot \binom{n + d - 1}{d} \stackrel{!}{\geq} \binom{n + (2 \cdot D + 1) + d - 1}{(2 \cdot D + 1) + d}.
    \end{equation}
    In \Cref{Tab: complexity estimates} we list our complexity estimates together with estimates for various lattice-based attacks.
    The complexities for lattice-based attacks have been computed via the lattice estimator tool\footnote{\url{https://github.com/malb/lattice-estimator}} by Albrecht et al.\ \cite{Albrecht-Hardness}.

    \begin{table}[H]
        \centering
        \caption{Bit complexity estimation for various attack strategies on Kyber768.
                 Complexity of lattice-based attacks are computed via the lattice estimator \cite{Albrecht-Hardness}.
                 For attacks where the lattice estimator provides estimations for multiple steps in an attack the most difficult step is shown in the table.
                 For Gr\"obner basis attacks, the proven complexity estimate is computed via \Cref{Equ: complexity estimate} and the Macaulay bound (\Cref{Cor: Macauly bound}).
                 The optimistic complexity estimate is computed via \Cref{Equ: refined complexity estimate}, \Cref{Th: refined complexity estimate} and the lowest achievable degree of regularity.
                 The lowest achievable complexity estimate is computed via \Cref{Equ: complexity estimate} with $\solvdeg_{DRL} \left( \mathcal{F}_\text{Kyber768} \right) \leq d_{\reg} \left( \mathcal{F}_\text{Kyber768} \right) + (2 \cdot D + 1) - 1$ (\Cref{Conj: regularity bound}).
                 Gr\"obner basis complexity estimates are computed with $\omega = 2$.
        }
        \label{Tab: Kyber768 complexity estimation}
        \resizebox{\linewidth}{!}{
            \begin{tabular}{ M{18mm} | M{10mm} | M{10mm} | M{10mm} | M{12mm} | M{12mm} | M{10mm} | M{12mm} || M{9mm} M{9mm} | M{9mm} M{9mm} | M{9mm} M{9mm} }
                \toprule
                Method  & BKW & USVP & BDD & BDD Hybrid & BDD MiTM Hybrid & Dual & Dual Hybrid & \multicolumn{2}{ M{18mm} | }{Proven complexity estimate} & \multicolumn{2}{ M{18mm} | }{Optimistic complexity estimate} & \multicolumn{2}{ M{18mm} }{Lowest achievable complexity estimate} \\
                \midrule

                Samples           & $2^{226}$ & $768$ & $768$ & $768$ & $768$ & $768$ & $768$ & $768$  & $768^4$   & $768$  & $768^4$   & $768$  & $768^4$   \\
                \hline
                Complexity (bits) & $239$     & $205$ & $201$ & $201$ & $357$ & $214$ & $206$ & $5554$ & $5581$    & $4717$ & $419$     & $1588$ & $203$    \\

                \hline

                Solving degree & \multicolumn{7}{ c || }{n.a.} & \multicolumn{2}{ c |}{$3077$} & \multicolumn{2}{ c | }{n.a.} & \multicolumn{2}{ c }{n.a.} \\
                \hline
                Lowest achievable degree of regularity & \multicolumn{7}{ c || }{n.a.} & \multicolumn{2}{ c |}{n.a.} & $232$ & $7$ & $232$ & $7$ \\

                \bottomrule
            \end{tabular}
        }
    \end{table}

    \section{Integrating Hints into LWE Polynomial Models}\label{Sec: hints}
    In two recent works Dachman-Soled et al.\ \cite{C:DDGR20,C:DGHK23} introduced a framework for cryptanalysis of LWE in the presence of side information.
    E.g., in presence of a side-channel the information can come from the power consumption, electromagnetic radiation, sound emission, etc.\ of a device.
    Once side information has been obtained it has to be modeled as mathematical hints.
    Dachman-Soled et al.\ categorize hints for LWE into four classes \cite[\S 1]{C:DDGR20}:
    \begin{itemize}
        \item Perfect hints: $\braket{\mathbf{s}, \mathbf{v}} = l \in \Fq$.\footnote{
            Dachman-Soled et al.\ \cite{C:DDGR20} considered perfect hints over $\mathbb{Z}^n$, our notion of perfect hint corresponds to their modular hint, where the modulus is the characteristic of $\Fq$.
            They made this distinction, because affine equations over $\mathbb{Z}^n$ and $\Fqn$ require different integration into lattice algorithms, see \cite[\S 4.1, 4.2]{C:DDGR20}.
            Though, for integration into polynomial systems perfect hints are always projected to $\Fq$.
            }

        \item Modular hints: $\braket{\mathbf{s}, \mathbf{v}} \equiv l \mod k$.

        \item Approximate hints: $\braket{\mathbf{s}, \mathbf{v}} + e_\sigma = l \in \Fq$.

        \item Short vector hints: $\mathbf{v} \in \Lambda$, where $\Lambda$ is the lattice associated to a LWE instance.
    \end{itemize}

    Dachman-Soled et al.\ \cite{C:DDGR20,C:DGHK23} then discuss how these hints can be incorporated into Distorted Bounded Distance Decoding (DBDD) problems and lattice reduction algorithms to attack LWE.
    For readers interested how such hints can be obtained in practice we refer to \cite[\S 4, 6]{C:DDGR20}.
    Except for short vector hints that do not involve the LWE secret, we can incorporate these hints into LWE polynomial models.

    Integrating a perfect hint is straight-forward since including an affine equation to the polynomial systems simply eliminates one variable.

    If we are given a modular hint, then in principle one can compute a subset $\Omega \in \Fq$ such that $\braket{\mathbf{s}, \mathbf{v}} - l \in \Omega$ (in $\Fq$).
    Hence, we can set up a new polynomial with roots in $\Omega$, substitute $\braket{\mathbf{s}, \mathbf{v}} - l$ into the polynomial and add it to the LWE polynomial system.
    Although this sounds simple, in practice the computation of $\Omega$ can be a challenge.
    In particular, if $\mathbf{s}$ and $\mathbf{v}$ can take all values in $\Fq^n$, then we expect the set $\Omega$ to be too big to improve Gr\"obner basis computations.
    On the other hand, if $\mathbf{s}, \mathbf{v} \in \{ 0, 1 \}^n$ and we have the modular equation $\braket{\mathbf{s}, \mathbf{v}} \equiv 1 \mod 2$, then only the odd numbers in the interval $\left[ 0, n \right]$ can be in $\Omega$, so the univariate polynomial with roots in $\Omega$ is of degree $\leq \ceil{\frac{n}{2}}$.

    More interesting are approximate hints.
    Such hints are obtained from noisy side-channel information.
    In case the probability distribution of $e_\sigma$ has smaller width than the one of the LWE error, then we can reduce the degree of a polynomial in the LWE polynomial system.
    Another class of hints that we interpret as approximate hints are Hamming weight hints.
    Suppose that the LWE secret entry $s_1$ is drawn from $D \subset \Fq$ and that we know the Hamming weight $H (s_1) = k$.
    Then we can add a univariate polynomial in $x_1$ to the LWE polynomial system whose roots are exactly the elements of $D$ of Hamming weight $k$.
    I.e., Hamming weight hints restrict the number of possible solutions.
    We illustrate this with an example.
    \begin{ex}\label{Ex: Kyber inspired}
        Let $q$ be a $16$ bit prime number, and let $(\mathbf{a}_i, b_i)_{1 \leq i \leq m} \subset \Fqn \times \Fq$ be a LWE sample generated with secret $\mathbf{s} \subset [-5, 5]^n$.
        As discussed in \Cref{Sec: small secret LWE}, for every variable $x_i$ we can add a polynomial of degree $11$ to the polynomial system to restrict the solutions to the interval.
        Suppose that $s_i$ is represented by a signed $16$ bit integer and that we learned its Hamming weight $H (s_i) = 2$, then $s_i \in \{ 3, 5 \}$ and we can replace the degree $11$ polynomial by a polynomial of degree $2$.
    \end{ex}

    Note that such Hamming weight biases can also persist if one opts for a more efficient memory representation of the secret entries.
    \begin{ex}\label{Ex: Dilithium inspired}
         Let $q$ be a $16$ bit prime number, and let $(\mathbf{a}_i, b_i)_{1 \leq i \leq m} \subset \Fqn \times \Fq$ be a LWE sample generated with secret $\mathbf{s} \subset [-2, 2]^n$.
         As discussed in \Cref{Sec: small secret LWE}, for every variable $x_i$ we can add a polynomial of degree $5$ to the polynomial system to restrict the solutions to the interval.
         Assume that the entries of $\mathbf{s}$ are stored as signed integers in the interval $\left[ -\frac{q}{2}, \frac{q}{2} \right]$, then
         \begin{itemize}
             \item if $H (s_i) = 0$, then $s_1 = 0$,

             \item if $H (s_i) = 1$, then $s_1 \in \{ 1, 2 \}$, and

             \item if $H (s_i) = 2$, then $s_1 \in \{ -1, -2 \}$.
         \end{itemize}
         So if one can learn the Hamming weight of $s_i$, then one either obtains a perfect hint or one can replace the degree $5$ polynomial by a degree $2$ polynomial.
    \end{ex}

    Moreover, modular and approximate hints can be combined in a hybrid manner.
    \begin{ex}\label{Ex: hints complexity}
        Let $q$ be a $16$ bit prime number, and let $(\mathbf{a}_i, b_i)_{1 \leq i \leq m} \subset \Fqn \times \Fq$ be a LWE sample generated with secret $\mathbf{s} \subset [-5, 5]^n$.
        Assume that the entries of $\mathbf{s}$ are stored as signed integers in the interval $\left[ -\frac{q}{2}, \frac{q}{2} \right]$.
        If $H (s_i) = 2$ and $s_i \equiv 1 \mod 3$, then $s_i \in \{ -2, 4 \}$.
        So we can replace the degree $11$ polynomial by a polynomial of degree $2$.
    \end{ex}

    In practice this can have devastating consequences.
    If we can reduce a small secret LWE instance to binary secret LWE or even worse to binary secret binary error LWE, then we expect to achieve a lower degree of regularity with less number of samples necessary compared to the plain polynomial system.
    We numerically showcase this in the following example.
    \begin{ex}\label{Ex: Kyber hints}
        Let $q$ be a $16$ bit prime number, assume that we are given small secret small error LWE over $\Fq^{256}$ whose secrets and error are drawn from $[-2, 2]$.
        Let $m = 256^\frac{3}{2}$ samples be given, and assume that we have enough Hamming weight hints for the secret and the error terms to transform the LWE polynomial system to either
        \begin{enumerate}[label=(\roman*)]
            \item binary secret LWE, or

            \item binary secret binary error LWE.
        \end{enumerate}
        In \Cref{Tab: complexity estimates} we record the least integer $d$ such that $d_{\reg} \left( \mathcal{F}_\text{LWE} \right) \leq D + d$ together with the optimistic complexity estimate from \Cref{Equ: refined complexity estimate} and the lowest achievable complexity estimate implied by \Cref{Equ: regularity lower bound} for various numbers of perfect hints.
    \end{ex}
    \clearpage 
    \begin{table}[H]
        \centering
        \caption{Complexity estimates for small secret small error LWE, binary secret LWE and binary secret binary error LWE over $\Fq^{256}$ with error polynomial degree $D = 5$ and $m = 256^\frac{3}{2}$.
                 The column $d$ lists the least integer such that $d_{\reg} \left( \mathcal{F}_\text{LWE} \right) \leq D + d$ for a given number of perfect hints.
                 The optimistic complexity estimate is computed via \Cref{Equ: refined complexity estimate} and the lowest achievable complexity estimate is computed via \Cref{Equ: complexity estimate} with $\solvdeg_{DRL} \left( \mathcal{F}_\text{LWE} \right) = d_{\reg} \left( \mathcal{F}_\text{LWE} \right) + D - 1$ where $D = 5, 2$ (\Cref{Conj: regularity bound}).}
        \label{Tab: complexity estimates}
        \resizebox{\textwidth}{!}{
            \begin{tabular}{ M{12mm} || M{5mm} | M{20mm} | M{25mm} || M{5mm} | M{20mm} | M{25mm} || M{5mm} | M{20mm} | M{25mm} }
                \toprule
                & \multicolumn{3}{ c || }{Small Secret Small Error LWE} & \multicolumn{3}{ c || }{Binary Secret LWE} & \multicolumn{3}{ c }{Binary Secret Binary Error LWE} \\
                & \multicolumn{3}{ c || }{$D = 5$} & \multicolumn{3}{ c || }{$D = 5$} & \multicolumn{3}{ c }{$D = 2$} \\
                \midrule
                Perfect hints & $d$ & Optimistic complexity estimate (bits) & Lowest achievable complexity estimate (bits) & $d$ & Optimistic complexity estimate (bits) & Lowest achievable complexity estimate (bits) &
                $d$ & Optimistic complexity estimate (bits) & Lowest achievable complexity estimate (bits) \\
                \midrule

                \multicolumn{ 10 }{ c }{$\omega = 2$} \\
                \midrule

                $0$   & $57$ & $1391$ & $481$ & $38$ & $1118$ & $370$ & $3$ & $237$ & $92$  \\
                $50$  & $45$ & $1122$ & $393$ & $30$ & $906$  & $303$ & $2$ & $188$ & $78$  \\
                $150$ & $22$ & $596$  & $221$ & $15$ & $499$  & $174$ & $1$ & $127$ & $59$  \\
                $190$ & $13$ & $387$  & $152$ & $8$  & $320$  & $117$ & $0$ & $80$  & $45$  \\

                \midrule
                \multicolumn{ 10 }{ c }{$\omega = 3$} \\
                \midrule

                $0$   & $57$ & $1731$ & $712$ & $38$ & $1389$ & $547$ & $3$ & $291$ & $131$ \\
                $50$  & $45$ & $1394$ & $580$ & $30$ & $1125$ & $336$ & $2$ & $229$ & $110$ \\
                $210$ & $8$  & $339$  & $165$ & $5$  & $290$  & $127$ & $0$ & $88$  & $56$  \\

                \bottomrule
            \end{tabular}
        }
    \end{table}

    \section{Discussion}
    In this paper we proved that any fully-determined LWE polynomial system is in generic coordinates.
    Therefore, bounds for the complexity of DRL Gr\"obner basis computations can be found via the Castelnuovo-Mumford regularity.
    In particular, this permits provable complexity estimates without relying on strong but unproven theoretical assumptions like semi-regularity \cite{Froeberg-Conjecture,Pardue-Generic}.

    We also demonstrated how the degree of regularity of a LWE polynomial system can be used to derive complexity estimates.
    Though, in practice one has to keep in mind that a degree of regularity computation usually requires a non-trivial Gr\"obner basis computation for the highest degree components.
    Hence, we interpret complexity bounds based on the lowest achievable degree of regularity as worst-case bounds from a designer's perspective that \emph{could} be achievable by an adversary.

    Based on the lowest achievable degree of regularity, we discussed that a conservative cryptanalyst should assume that Gaussian elimination on a single Macaulay matrix in the degree of regularity is sufficient to solve Search-LWE.

    Moreover, we discussed how side information can be incorporated into LWE polynomial systems, and we showcased how it can affect the complexity of Gr\"obner basis computations.

    Overall, we have presented a new framework to aid algebraic cryptanalysis for LWE-based cryptosystems under minimal theoretical assumptions on the polynomial system.

    \subsubsection{\ackname}
    The author would like to thank the anonymous reviewers at Eurocrypt 2024 for their valuable comments and helpful suggestions which improved both the quality and presentation of the paper.
    The author would like to thank Prof.\ Elisabeth Oswald at Alpen-Adria-Universit\"at Klagenfurt for her suggestion to study LWE polynomial system as well as the hints framework of Dachman-Soled et al.
    Matthias Steiner has been supported by the European Research Council (ERC) under the European Union's Horizon 2020 research and innovation program (grant agreement No.\ 725042).

    \appendix
    \section{Semaev \& Tenti's Probability Analysis}\label{Sec: Appendix}
    Let $P = \Fq [x_1, \dots, x_n]$ and $R = P / \left( x_1^q, \dots, x_n^q \right)$, and let $\mathcal{F} = \{ f_1, \dots, f_m\} \subset R$ be a homogeneous polynomial system such $\degree{f_i} = D$ for all $i$.
    Additionally, we define the notation
    \begin{equation}
        l_e (n, d) = \abs{\left\{ (a_1, \dots, a_n) \in \mathbb{Z}^n \; \middle\vert \; 0 \leq a_i < e,\ \sum_{i = 1}^{n} a_i = d \right\}}.
    \end{equation}
    Moreover, if we refer to the degree of regularity in the following paragraph, then we implicitly mean the extension of the degree of regularity to $R$ (see \cite[Definition~3.48, Theorem~3.53]{Tenti-Overdetermined}).

    Semaev \& Tenti analyzed the probability that a uniformly random polynomial system $\mathcal{F}$ achieves a certain degree of regularity.
    In particular, they proved that for $D > d > 0$ and $m \geq \frac{l_q (n, D + d)}{l_q (n, d)}$ the probability that $d_{\reg} \left( \mathcal{F} \right) \leq D + d$ converges to $1$ for $n \to \infty$ \cite[Theorem~1.1]{Semaev-Complexity}.
    For $q = D = 2$ Tenti also provided an explicit probability in his PhD thesis \cite[Theorem~4.2]{Tenti-Overdetermined}
    \begin{equation}
        \prob \Big[ d_{\reg} \left( \mathcal{F} \right) \leq 3 \Big] \geq 1 - \sum_{v = 0}^{n - 1} q^{\binom{n - v}{3} + (n - v + 1) \cdot v - (n - v) \cdot m}.
    \end{equation}

    For LWE polynomial systems we encounter a very similar scenario, since in all scenarios in \Cref{Sec: general LWE,Sec: small secret LWE,Sec: small error LWE} we can find a set of univariate polynomials $F = \big( f (x_1), \dots, f (x_n) \big)$, where $f$ is univariate and $\degree{f} = e \geq 2$, that restricts the number of possible solutions.
    Moreover, in all our scenarios $f (x_i) \mid x_i^q - x_i$.
    Thus, to analyze the degree of regularity analog to Semaev \& Tenti we first construct the polynomials $F$, and then we replace the polynomials in $\mathcal{F}_\text{LWE}$ by their remainders modulo $F$.
    Finally, for the degree of regularity only the highest degree components matter, so we can restrict the analysis to $\mathcal{F}_\text{LWE}^\topcomp \in \Fq [x_1, \dots, x_n] / \left( x_1^e, \dots, x_n^e \right)$.

    The proof of Semaev \& Tenti is combinatorial, i.e.\ in principle it does not depend on $e = q$ nor the characteristic of the finite field.
    Therefore, we expect that their results can be generalized to arbitrary rings $\Fq [x_1, \dots, x_n] / \left( x_1^e, \dots, x_n^e \right)$.

    Instead of repeating their analysis now, we will investigate a simpler problem.
    Can we consider the highest degree components of LWE polynomial systems as uniformly distributed in $\Fq [x_1, \dots, x_n]$?

    \subsection{On The Distance Of LWE Highest Degree Components To The Uniform Distribution}\label{Sec: distance to uniform distrubition}
    The highest degree components of LWE polynomials are of the form $\big( \! \braket{\mathbf{a}_i, \mathbf{x}} \! \big)^d$, for some uniform and independent distributed $\mathbf{a}_i \in \Fqn$.
    In this section we compute the distance of the coefficient distribution of $\big( \! \braket{\mathbf{a}_i, \mathbf{x}} \! \big)^d$ to the uniform distribution over $\Fq^{\binom{n + d - 1}{d}}$.
    Let us first recall the notion of statistical distance also known as total variation distance.
    \begin{defn}[{\cite[\S 4.1]{Levin-Markov}}]\label{Def: statistical distance}
        Let $\mu$ and $\nu$ be two probability distributions on a finite set $\Omega$.
        The statistical distance (or total variation distance) between $\mu$ and $\nu$ is defined as
        \begin{equation*}
            \dtv \left( \mu, \nu \right) = \frac{1}{2} \cdot \sum_{x \in \Omega} \left\vert \mu (x) - \nu (x) \right\vert.
        \end{equation*}
    \end{defn}

    The following identity will be useful to compute the statistical distance.
    \begin{lem}[{\cite[Remark~4.3]{Levin-Markov}}]\label{Lem: formula for statistical distance}
        Let $\mu$ and $\nu$ be two probability distributions on a finite set $\Omega$.
        Then
        \begin{equation*}
            \dtv \left( \mu, \nu \right) = \sum_{\substack{x \in \Omega, \\ \mu (x) \geq \nu (x)}} \mu (x) - \nu (x).
        \end{equation*}
    \end{lem}

    The following lemma is an easy consequence of the previous lemma.
    \begin{lem}
        Let $\mu$ and $\nu$ be two probability distributions on a finite set $\Omega$, and assume that $\dtv \left( \mu, \nu \right) \leq \epsilon$ for some $\epsilon > 0$.
        Then
        \begin{align*}
            \nu (x) - \epsilon \leq &\mu (x) \leq \nu (x) + \epsilon, \\
            \mu (x) - \epsilon \leq &\nu (x) \leq \mu (x) + \epsilon.
        \end{align*}
    \end{lem}
    \begin{proof}
        If $\mu (x) - \nu (x) \geq 0$, then by \Cref{Lem: formula for statistical distance} $\mu (x) - \nu (x) \leq \sum_{\substack{x \in \Omega, \\ \mu (x) \geq \nu (x)}} \mu (x) - \nu (x) = \dtv \left( \mu, \nu \right) \leq \epsilon$.
        If $\nu (x) - \mu (x) \geq 0$, then $\nu (x) - \mu (x) \leq \sum_{\substack{x \in \Omega, \\ \nu (x) \geq \mu (x)}} \nu (x) - \mu (x) = \dtv \left( \mu, \nu \right) \leq \epsilon$.
        By combining these two inequalities we derive the claims. \qed
    \end{proof}

    Hence, if some property holds for the distribution $\nu$, then in principle one can extend this property to the distribution $\mu$ up to some error term that depends on $\epsilon$.

    Let us now return to $\big( \! \braket{\mathbf{a}_i, \mathbf{x}} \! \big)^d$, for ease of notation we abbreviate $N_{n, d} = \binom{n + d - 1}{d}$.
    We consider the function
    \begin{equation}
        \begin{split}
            \phi_{n, d}: \Fqn &\to \Fq^{N_{n, d}}, \\
            (a_1, \dots, a_n) & \mapsto \left( \frac{d!}{i_1! \cdots i_n!} \cdot a_1^{i_1} \cdots a_n^{i_n} \right)_{\subalign{0 \leq &i_j \leq d, \\ \sum_{j = 1}^{n} &i_j = d}}.
        \end{split}
    \end{equation}
    Obviously, $\phi$ maps $\mathbf{a} \in \Fqn$ to the coefficient vector of $(\braket{\mathbf{a}, \mathbf{x}})^d$.
    Moreover, we assume that $\frac{n!}{i_1! \cdots i_n!} \neq 0$ in $\Fq$ for all possible $(i_1, \dots, i_n)$.
    This condition is for example satisfied if $q$ is prime and $d < q$.
    Then
    \begin{equation}
        \mu (\mathbf{b}) = \prob \left[ \mathbf{b} \in \imag (\phi_{n, d}) \right] = \frac{\abs{\phi_{n, d}^{-1} (\mathbf{b})}}{q^n}.
    \end{equation}
    In particular, if $\abs{\phi_{n, d}^{-1} (\mathbf{b})} \neq 0$, then $\mu (\mathbf{b}) \geq \frac{1}{q^n}$.
    Also, for $d \geq 1$ we have that $N_{n, d} \geq n$.
    Then with \Cref{Lem: formula for statistical distance} we have that
    \begin{align}
        \dtv \left( \mu, \frac{1}{q^{N_{n, d}}} \right)
        &= \sum_{\substack{\mathbf{b} \in \Fq^{N_{n, d}}, \\ \mu (\mathbf{b}) \geq \frac{1}{q^{N_{n, d}}}}} \mu (\mathbf{b}) - \frac{1}{q^{N_{n, d}}} \\
        &= \sum_{\mathbf{b} \in \imag (\phi_{n, d})} \frac{\abs{\phi_{n, d}^{-1} (\mathbf{b})}}{q^n} - \frac{1}{q^{N_{n, d}}} \\
        &= 1 - \frac{\abs{\imag (\phi_{n, d})}}{q^{N_{n, d}}} \\
        &\geq 1 - \frac{q^n}{q^{N_{n, d}}}.
    \end{align}
    Obviously, the last expression has limit $1$ for $d \to \infty$.
    Now let $\mathbf{a} \in \Fqn$ be uniformly random and let $\mathbf{b} \in \Fq^{N_{n, d}}$, then
    \begin{equation}
        \prob \left[ \phi_{n, d} (\mathbf{a}) = \mathbf{b} \right] = \sum_{\mathbf{c} \in \Fqn} \prob \left[ \mathbf{c} \right] \cdot \prob \left[ \phi_{n, d} (\mathbf{a}) = \mathbf{b} \mid \mathbf{a} = \mathbf{c} \right].
    \end{equation}
    Note that $\prob \left[ \phi_{n, d} (\mathbf{a}) = \mathbf{b} \mid \mathbf{a} = \mathbf{c} \right] \in \{ 0, 1 \}$, and it is equal to $1$ exactly $\abs{\phi_{n, d}^{-1} (\mathbf{b})}$-many times.
    Therefore,
    \begin{equation}
        \prob \left[ \phi_{n, d} (\mathbf{a}) = \mathbf{b} \right] = \prob \left[ \mathbf{b} \in \imag (\phi_{n, d}) \right].
    \end{equation}
    Hence, in general we consider the distribution of $\phi_{n, d} (\mathbf{a})$, where $\mathbf{a} \in \Fqn$ is uniformly random, to be far from the uniform distribution over $\Fq^{N_{n, d}}$.
    For example for $d = 2$ we have that
    \begin{equation}
        \dtv \left( \mu, \frac{1}{q^{N_{n, 2}}} \right) \geq 1 - \frac{q^n}{q^\frac{(n + 1) \cdot n}{2}} = 1 - q^\frac{-n^2 + n}{2} \geq \frac{1}{2},
    \end{equation}
    where the last inequality follows by $n, q \geq 2$.

    Hence, probability estimations for uniformly distributed highest degree components, like the one of Semaev \& Tenti, are not applicable to LWE polynomial systems.

    \bibliographystyle{splncs04.bst}
    \bibliography{abbrev0.bib,crypto.bib,literature.bib}

\end{document}